\documentclass{article}

\usepackage{microtype}
\usepackage{amsmath}
\usepackage{amssymb}
\usepackage{amsthm}
\usepackage{enumitem}
\usepackage{graphicx}
\usepackage{multicol}
\usepackage{booktabs}
\usepackage{xcolor}

\usepackage{sectsty}

\newtheorem{definition}{Definition}
\newtheorem{theorem}{Theorem}
\newtheorem{corollary}{Corollary}
\newtheorem{prop}{Proposition}
\newtheorem{lemma}{Lemma}

\DeclareMathOperator{\tw}{tw}
\DeclareMathOperator{\poly}{poly}

\title{Local treewidth of random and noisy graphs with applications to stopping contagion in networks}
\author{Hermish Mehta \and Daniel Reichman}
\begin{document}

\maketitle
\begin{abstract}
    We study the notion of local treewidth in sparse random graphs: the maximum treewidth over all $k$-vertex subgraphs of an $n$-vertex graph. When $k$ is not too large, we give nearly tight bounds for this local treewidth parameter; we also derive tight bounds for the local treewidth of noisy trees, trees where every non-edge is added independently with small probability. We apply our upper bounds on the local treewidth to obtain fixed parameter tractable algorithms (on random graphs and noisy trees) for edge-removal problems centered around containing a contagious process evolving over a network. In these problems, our main parameter of study is $k$, the number of initially ``infected'' vertices in the network. For the random graph models we consider and a certain range of parameters the running time of our algorithms on $n$-vertex graphs is $2^{o(k)}\poly(n)$, improving upon the $2^{\Omega(k)}\poly(n)$ performance of the best-known algorithms designed for worst-case instances of these edge deletion problems.
\end{abstract}


\section{Introduction}

Treewidth is a graph-theoretic parameter that measures the resemblance of a graph to a tree.
We begin by recalling the definition of treewidth. 

\begin{definition}[Tree Decomposition]
	A tree decomposition of a graph $G = (V, E)$ is a pair $(T, X)$, where $X$ is a collection of subsets of $V$, called bags, and $T$ a tree on vertices $X$ satisfying the properties below:
\end{definition}

\begin{enumerate}[noitemsep]
	\item The union of all sets $X_i \in X$ is $V$.
	\item For all edges $(u, v) \in E$, there exists some bag $X_i$ which contains both $u$ and $v$.
	\item If both $X_i$ and $X_j$ contain some vertex $u \in V$, then all bags $X_k$ on the unique path between $X_i$ and $X_j$ in $T$ also contain $u$.
\end{enumerate}

\begin{definition}[Treewidth]
	The width of a tree decomposition $(T, X)$ is one less than cardinality of the largest bag. More formally, we can express this as 
	\begin{align*}
		\max_i |X_i| - 1.
	\end{align*}
	The treewidth of a graph $G = (V, E)$ is the minimum width among all tree decompositions of $G$.
\end{definition}
Many graph-theoretic problems that are NP-hard admit polynomial-time algorithms on graph families whose treewidth is sufficiently slowly growing as a function of the number of vertices~\cite{kloks1994treewidth}. There is vast literature concerned with finding methods to relate the treewidth of graphs to other well-studied combinatorial parameters and leveraging this to devise efficient algorithms for algorithmic problems in graphs with constant or logarithmic treewidth~\cite{cygan2015parameterized}. An excellent introduction to the concept of treewidth as well as brief survey of the work of Robertson and Seymour in establishing this concept can be found in Chapter 12 of~\cite{diestel2005graph}.

These treewidth-based algorithmic methods, however, have historically found limited applicability in random graphs. Sparse random graphs $G(n, d/n)$ where every edge occurs independently with probability $d/n$, for some $d > 1$, exhibit striking contrast between their local and global properties---and this contrast is apparent when looking at treewidth. Locally, these graphs appear tree-like with high probability\footnote{Given a random graph model, we say an event happens with high probability if it occurs with probability tending to $1$ as $n$ tends to infinity.} (w.h.p.): the ball of radius $O(\log_d n)$ around every vertex looks like a tree plus a constant number of additional edges. Globally, however, these graphs have w.h.p. treewidth $\Omega(n)$. For example, the super-critical random graph $G(n, \frac{1+\delta}{n})$ has w.h.p. treewidth $\Omega(n)$~\cite{do2022note,perarnau2014tree,lee2012rank}. As a result of this global property, conventional techniques used to exploit low treewidth to derive efficient algorithms do not apply directly for random graphs.

In this paper, we take advantage of the local tree-like structure of random graphs by analyzing the \emph{local} behavior of treewidth in random graphs. Central to our approach is the following definition.

\begin{definition}[Local Treewidth]\label{def:local}
	Let $G$ be an undirected $n$-vertex graph. Given $k \leq n$ we denote by $t_k(G)$ the largest treewidth of a subgraph of cardinality $k$ of $G$.
\end{definition}

In words, the local treewidth of an $n$-vertex graph, with locality parameter $k$, is the maximum possible treewidth across all subgraphs of size $k$. We study two models of random graphs, starting with the familiar binomial random graph $G(n, p)$. While the binomial random graph $G(n,p)$ lacks many of the characteristics of empirically observed networks such as skewed degree distributions, studying algorithmic problems on random graphs can nevertheless lead to interesting algorithms. 

\begin{definition}[Noisy Trees]\label{def:noisy}
Let $T$ be an $n$ vertex tree. The noisy tree $T'$ obtained from $T$ is a random graph model where
every non edge of $T$ is added to $T$ independently, with probability $1/n$. 
\end{definition}

Here we assume $p=1/n$ for convenience; all our results regarding noisy trees also hold when the perturbation probability $p$ satisfies $p=\epsilon/n$ for $\epsilon <1$. Noisy trees are related to small world models of random networks~\cite{newman1999renormalization,newman2011scaling}, where adding a few random edges to a graph of high diameter such as a path results with a graph of logarithmic diameter w.h.p.~\cite{krivelevich2015smoothed}.

Below, we give an informal description of the concepts we study and sketch our main results; we defer discussion of formal results until Section~\ref{sec:formal} and later in the paper. Our main result is a nearly tight bound holding w.h.p. for the maximum treewidth of a $k$-vertex subgraph of $G(n, p)$ assuming $k \leq n^{1-\epsilon}$ for $\epsilon \in (0, 1)$ and $p=d/n$ where $d > 1$. In the notation introduced earlier, this provides a bound for $t_k(G)$. Assuming $k \leq n^{\epsilon}$ for a sufficiently small $\epsilon$ we obtain nearly tight bounds for the local treewidth of noisy trees as well.
 
Our upper bounds on the local treewidth are motivated by algorithmic problems related to containing the spread of a contagious process over undirected graphs by deleting edges. We focus on the bootstrap percolation contagious process (Definition~\ref{def:boot}) where there is a set of initially infected vertices and noninfected vertices are infected if they have at least $r \geq 2$ neighbors and consider two edge-removal problems: \emph{Stopping Contagion} and \emph{Minimizing Contagion}. Informally, in stopping contagion we are given a subset of infected nodes $A$ and seek to remove a minimal number of edges to ensure a ``protected" subset of vertices $B$ (disjoint from $A$) are not infected from $A$. In minimizing contagion we wish to ensure at most $m$ additional vertices are infected from $A$ for a target value $m$ by deleting a minimal number of edges. Such edge removal problems might arise, among other applications~\cite{enright2018deleting,enright2021deleting}, in railways and air routes, where the goal might be to prevent spread while also minimizing interference to transportation. In this context, edge deletion may correspond to removing a transportation link altogether or introducing  special requirements (such as costly checks) to people between the the endpoints. Edge removal can be also viewed as a \emph{social distancing}  measure to control an epidemic outbreak~\cite{babay2022controlling}. One can also study the problem of removing \emph{vertices} to control the spread of an epidemic which is related to vaccinations: making nodes immune to infection and removing them from the network~\cite{sambaturu2020designing}.

We design algorithms for stopping and minimizing contagion for random graphs and noisy trees.
Note that our algorithms do not achieve polynomial time, even for $k$ that is poly-logarithmic in $n$; whether there exists a polynomial time algorithm for minimizing contagion and stopping contagion in $G(n, p)$ for every value of $k$ is an open question. Nonetheless, the dependency of our algorithm on $k$ is better (assuming $k \leq n^{\epsilon}$ for an appropriate constant $\epsilon>0$) than the dependency of $k$ in the running time of the best known algorithms for minimizing contagion\footnote{We are not aware of previous algorithms for the stopping contagion problem.} in the worst case~\cite{cordasco2021parameterized}. Please see Subsection~\ref{sub:algorithms} for details. 

Our algorithms are based on the following three observations:

\begin{enumerate}[noitemsep]
\item The local treewidth of binomial random graphs and noisy trees is sublinear in $k$. 
\item There exist fast algorithms for minimizing and stopping contagion in graphs of bounded treewidth. 
\item The set of seeds $A$ has what we call the \emph{bounded spread} property: w.h.p. at most $c|A|$ additional vertices are infected from $A$ for some constant\footnote{For $G(n,d/n)$, our constant $c:=c(d)$ is a function of $d$. When $d$ is a constant independent of $n$ so is $c$.} $c$. Bounded spread allows us to solve minimizing contagion and stopping contagion on subgraphs that have small (sublinear in $k$) treewidth.
\end{enumerate}

For the sake of brevity and readability we focus on \emph{edge} deletion problems. We note that our algorithms can be easily adapted for the analogous problems of minimizing and stopping contagion by deleting \emph{vertices}
rather than edges. The reason is that our algorithms for minimizing/stopping contagion on bounded treewidth graphs work (with the same asymptotic running time guarantees) for vertex deletion problems.
Combining algorithms for bounded treewidth with the bounded spread property as well the upper bound on the local treewidth yields algorithms for the vertex deletion versions of minimizing and stopping contagion. 

Our main contribution is studying the concept of local treewdith for random graphs and connecting it to algorithmic problems involving stopping contagion in networks. Our calculations are standard and the contribution is conceptual rather than introducing a new technique. Our results for noisy trees regarding the bounded spread property are interesting as in contrast to other infection models considered in the literature~\cite{becchetti2021sharp}, the influence of adding random ``long range'' edges to the total spread of a seed set with at most $n^{1-\epsilon}$ vertices is minor in the sense that w.h.p. the total spread increases as a result only by a constant multiplicative factor.

\section{Our results}\label{sec:formal}
\subsection{Local treewidth bounds}
Recall we define the local treewidth of a graph $G$, denoted $t_k(G)$, to be the greatest treewidth among along subgraphs of size $k$. Trivially, for any graph with at least one edge and $k \leq n$, $1 \leq t_k(G) \leq k$.

Consider as an illustrative example the random graph $G = G(n, 1/2)$: with high probability, $t_k(G)=\Omega(k)$ for all values of $k$. For $k \leq 1.9 \log n$ this follows as there is a clique of size $k$ in $G$ w.h.p. For $k > 1.9 \log n$ this follows as a randomly chosen subset of size $k$ has, with high probability, minimum degree $\Omega(k)$, and a graph with treewidth $r$ has a vertex of degree at most $r$.

We can now state our bounds for $t_k$ in the random graph models we consider. From here onward, $\epsilon>0$ is taken to be a positive constant in $(0,1 )$. We give somewhat compressed statements; reference to the full Theorems are provided throughout this section.

\begin{theorem}
Let $G = G(n, p)$ with $p=d/n$ and $k \leq n^{1-\epsilon}$. Then, with high probability:
\begin{align*}
    t_k(G) \leq 3+ O\left(\frac{k \log d}{\log n}\right).
\end{align*}
\end{theorem}

Since it is always the case that $t_k(G)\leq k$, the upper bound in the Theorem above becomes trivial if $d \geq n^{\Omega(1)}$. Also observe that the Theorem does not hold for arbitrary $k \leq n$, as for $k = n, t_k(G)=\Omega(n)$ w.h.p. In terms of lower bounds, we have  the following:

\begin{theorem}
Suppose $p=d/n$ and $d>1+\delta$ where $\delta>0$ is a constant (not depending on $n$). Suppose $k \leq O(n/\log n)$; then, w.h.p.
\begin{align*}
    t_k(G) \geq \Omega\left(\frac{k}{\log n}\right).
\end{align*}
\end{theorem}

More details can be found in Section~\ref{sec:random}. Our upper and lower bounds for the local treewidth of $G(n, d/n)$ also extend to the random $d$-regular graph $G(n, d)$--details can be found in Subsection~\ref{sec:random-reg}.

For noisy trees, we have the following results.

\begin{theorem}
Let $T$ be an $n$-vertex tree with maximum degree $\Delta$. Let $T'$ be a noisy tree obtained from $T$. Then w.h.p.
\begin{align*}
    t_k(T') \leq 3 + O\left(\frac{k( \log k+\log \Delta)}{\log n}\right).
\end{align*}
\end{theorem}

Observe that the upper bound in the Theorem is trivial if $k, \Delta$ are $n^{\Omega(1)}$. As a result, in our proofs we will assume $k, \Delta \leq n^{\epsilon},$ for sufficiently small $\epsilon>0$. Our results can be extended to the case where each non-edge is added with probability $c/n$ for $c>1$. Similar ideas (which are omitted) yield the upper bound:

\begin{align*}
    t_k(T') \leq 3 + O\left(\frac{k( \log k+\log \Delta+\log c)}{\log n}\right).
\end{align*}

We also provide a lower bound, showing that up to the $\log k, \log \Delta$ terms, the upper bound above is tight. Namely, the noisy path has w.h.p. local treewidth of order $\Omega(k/ \log n)$. For more details on the lower and upper bounds please see Section~\ref{sec:noisy}.


\subsection{Contagious process and edge deletion problems}\label{sub:contagion}
The local treewidth results outlined above prove useful in the context of two edge deletion problems we study. These problems arise when considering the evolution of a contagious processes over an undirected graph.

We focus on the $r$-neighbor bootstrap percolation model \cite{chalupa1979bootstrap}.
\begin{definition}\label{def:boot}
In $r$-neighbor bootstrap percolation we are given an undirected graph $G=(V,E)$ and an integer \emph{threshold} $r\geq 1$. Every vertex is either \emph{active} (we also use the term infected) or \emph{inactive}; a set of vertices composed entirely of active vertices is called active. Initially, a set of vertices called \emph{seeds}, $A_0$, is activated. A contagious process evolves in discrete steps, where for integral $i > 0$,

\begin{align*}
    A_i=A_{i-1}\cup \{v \in V:|N(v)\cap A_{i-1}|\geq r\}.
\end{align*}

Here, $N(v)$ is the set of neighbors of $v$. In words, a vertex becomes active in a given step if it has at least $r$ active neighbors. 
An active vertex remains active throught the process and cannot become inactive. Set

\begin{align*}
    \langle A_0 \rangle=\bigcup_i A_i.
\end{align*}

The set $\langle A_0 \rangle$ is the set of nodes that eventually get infected from $A_0$ in $G$. Clearly, $\langle A_0 \rangle$ depends on the graph $G$, so we sometimes write $\langle A_0 \rangle_G$ to call attention to the underlying graph. We say a vertex $v \in V$ gets \emph{activated} or \emph{infected} from a set of seeds $A_0$ if $v \in \langle A_0 \rangle$. 
\end{definition}

It is straightforward to extend this definition to the case where every vertex $v$ has its own threshold $t(v)$ and a vertex is infected only if it has at least $t(v)$ active neighbors at some point. As is customary in bootstrap percolation models, we usually assume that all thresholds are larger than $1$. Now, given a network with an evolving contagious process, we introduce the stopping contagion problem:

\begin{definition}[Stopping Contagion]
    In the stopping contagion problem, we are given as input a graph $G=(V,E)$ along with two disjoints sets of vertices, $A, B \subseteq V$. Given that the seed set is $A$, the goal is to compute the minimum number of edge deletions necessary to ensure that no vertices from $B$ are infected. In other words, we want to make sure $\langle A \rangle_{G'} \cap B=\phi$, where $G'$ is the graph obtained from $G$ after edge deletions. Given an additional target parameter, $\ell$, the corresponding decision problem asks whether it is possible to ensure no vertices from $B$ are infected by deleting at most $\ell$ edges.
\end{definition}

Next we consider the setting where given a set of infected nodes, we want to remove the minimal number of edges to ensure no more than $s$ additional vertices are infected.

\begin{definition}
    In the minimizing contagion problem, we are given a graph $G=(V,E)$, a subset of vertices $A \subseteq V$ and a parameter $s$. Given that the seed set is $A$, we want to compute the minimum number of edge deletions required to ensure at most $s$ vertices in $V \setminus A$ are infected. If $G'$ is the graph obtained from $G$ by edge deletions, then this condition is equivalent to requiring $|\langle A \rangle_{G'}| \leq |A|+ s$. In the decision problem, we want to decide if it is possible to ensure $ |\langle A \rangle_{G'}| \leq |A|+ s$ with at most $\ell$ edge deletions.
\end{definition}

Both stopping contagion and minimizing contagion are NP-complete, and stopping contagion remains NP-hard even if $|A|=2$ and $|B|=1$. For complete proofs, please refer to the Appendix.
\subsection{Algorithmic results}\label{sub:algorithms}
For minimizing contagion, current algorithmic ideas~\cite{cordasco2021parameterized} can be used to prove that if $|A|$ and the optimal solution are of size $O(k)$ the problem can be solved in time $2^{O(k)}\poly(n)$ on $n$-vertex graphs. No such algorithm, parameterized by $|A|$ and the size of the optimal solution, is known for stopping contagion. Using our upper bounds for local treewidth, however, we can prove:

\begin{theorem}
Let $\epsilon$ be a constant in $(0,1)$. Suppose that $k \leq n^{1-\epsilon}$ and that every vertex has threshold greater than $1$. Let $G:=G(n,p)$ where $p=d/n$. Assuming $d$ is a constant, we have that w.h.p. both minimizing contagion and stopping contagion can be solved in $G$ in time $2^{o(k)}\poly(n)$.
\end{theorem}

\begin{theorem}
Suppose that $k \leq n^{\epsilon}$ for sufficiently small $\epsilon \in (0,1)$ and that every vertex has a threshold greater than $1$. Let $T'$ be a noisy tree where the base tree $T$ has maximum degree $\Delta=O(1)$. Then w.h.p. both minimizing contagion and stopping contagion can be solved in $T'$ in time $2^{o(k)}\poly(n)$.
\end{theorem}

We stress that the set $A$ of seeds can be chosen in arbitrary way. In particular, an adversary can pick $A$ \emph{after} the random edges in our graph models have been chosen.  

The dependence of the running time on $n, k, d$ and $\Delta$ can be made explicit: for precise statements, please see Section~\ref{sec:algorithms}. Algorithms for grids and planar graphs are presented in Section~\ref{sec:algorithms} as well.

For our purpose, to translate local treewidth bounds to algorithmic results, we need an algorithm for solving stopping contagion and minimizing contagion on graphs of low treewidth. We provide such an algorithm that runs in exponential time in the treewidth, assuming the maximum degree is not too large, using ideas from~\cite{cordasco2021parameterized}. More details can be found in Section~\ref{sec:treewidth}.

\subsection{Our techniques}

Our upper bounds for the local treewidth build on a simple ``edge excess principle": A $k$-vertex connected graph with $k+r$ edges has treewidth at most $r+1$. As the treewidth of a set of connected components is the maximum treewidth of a component, it suffices to analyze the number of edges in connected subgraphs of the random graphs we study. For $G(n, p)$ this is straightforward, but for noisy trees it is somewhat more involved. We find it easier to first analyze the edge excess of connected subgraphs, before considering connecting edges that allow us bound the excess of arbitrary subgraphs.

A key component in our lower bound is the simple fact that if $H$ is a minor of $G$ then $\tw(G) \geq \tw(H)$. Therefore it suffices to prove the existence of large treewidth subgraphs that are minors w.h.p. of random graphs and noisy trees. 
Recall that an $n$-vertex graph is called an $\alpha$-\emph{expander} if there exists $\alpha \in (0,1)$ such that every subset $S$ of vertices with at most $n/2$ vertices has at least $\alpha |S|$ neighbors not in $S$. We use the fact~\cite{krivelevich2019expanders} that for any graph $H$ with $k$ vertices and edges,
assuming $k=O(n/\log n)$ an $n$-vertex expander $G$ has an embedding\footnote{See Subsection~\ref{sec:def} for further details on minor-theoretic concepts we use.} of $H$ as a minor in $G$. Furthermore, every connected subgraph of $G$ corresponding to a vertex in $H$ is of size $O(\log n)$.  
The lower bound then follows as it is known that $G(n,\frac{1+\delta}{n})$ contains with high probability ~\cite{krivelevich2018finding,krivelevich2019expanders} a subgraph with $\Omega(n)$ vertices that is an $\alpha$-expander for an appropriate choice $\alpha$. Similar ideas are used to prove the existence of large minors with linear treewidth in the noisy trees (e.g., the noisy path).

Our algorithms for minimizing contagion and stopping contagion in graphs of bounded treewidth build on techniques designed to exploit the tree-like nature of low treewidth graphs, sharing similarities to algorithms for target set selection in ~\cite{ben2011treewidth}, where target set selection is the problem of finding a minimal set that infects an entire graph under the bootstrap percolation model. More directly, our problem resembles the Influence Diffusion Minimization (IDM) studied in \cite{cordasco2021parameterized}, where the goal is to minimize the spread of the $r$-neighbor bootstrap percolation process by preventing spread through vertices. After subdividing edges, minimizing contagion essentially reduces to IDM, albeit with additional restrictions on the vertices we can immunize (only vertices that belong to the "middle" of a subdivided edge can be deleted); we therefore solve a generalization of the IDM problem, see Definition~\ref{def:GIDM}, and use this to provide efficient algorithms for the problems we care about.

At a high-level, our algorithm works by solving the stopping contagion recursively on subgraphs and then combining these solutions via dynamic-programming until we have a solution for the whole graph. To combine subproblems successfully, at each step we explicitly compute solutions for all possible states of vertices in a bag. While this could take exponential time in general, this approach provides an efficient algorithm in graphs with bounded treewidth.

Our proof of bounded spread in noisy trees builds works by proving that small subsets of such trees contain few edges~\cite{coja2014contagious,feige2017contagious}. Since every non seed vertex needs at least two vertices to get infected, small contagious sets require small subsets that contain too many edges. Therefore, one can prove that small sets of seeds cannot infect too many vertices; the proof of small trees' local sparsity is similar to the proof that w.h.p. such noisy trees have small local treewidth. 

\subsection{Related work}

While the idea to remove edges or vertices to contain an epidemic has been studied before~\cite{ren2019generalized,braunstein2016network,aspnes2006inoculation}, most of these works focus using edge or vertex deletions that break the graph to connected components of sublinear (or even constant) size~\cite{enright2018deleting,ren2019generalized,braunstein2016network}. Recently approximation algorithms for edge deletion problems that arise in controlling epidemics has been studied in~\cite{babay2022controlling} for the SIR epidemic model. In particular,~\cite{babay2022controlling} studies the problem of deleting a set of edges of weight at most $B$ that minimizes the set of infected nodes after edges deletions. All these works consider a different contagion model from the $r\geq 2$ bootstrap percolation model studied here.

Bootstrap percolation was first introduced by statistical physicists~\cite{chalupa1979bootstrap} and has been studied on a variety of graphs such as grids~\cite{balogh2012sharp}, hypercubes~\cite{balogh2006bootstrap,morrison2018extremal}, random graphs~\cite{feige2017contagious}, graphs with a power law degree distribution~\cite{amini2014bootstrap,schoenebeck2016complex}, Kleinberg's small world model~\cite{ebrahimi2015complex} and trees~\cite{riedl2012largest}.

The fixed parameter tractability of minimizing contagion with respect to \emph{vertex} deletions, as opposed to edge deletions, has been thoroughly investigated with respect to various parameters such as the maximum degree, treewidth, and the size of the seed set $k$ in~\cite{cordasco2021parameterized}. The authors of ~\cite{cordasco2021parameterized} present efficient algorithms for minimizing contagion for graphs of bounded maximum degree and treewidth. With respect to $k$, using ideas from FPT algorithms for cut problems~\cite{fomin2013parameterized}, they give a $2^{k+\ell}\poly(n)$ algorithm for the case where the set of seeds is of size $k$ and there is a solution of size $\ell$ to the problem. Their algorithm can be easily adapted to the case of edge deletions: see Theorem~\ref{thm:exp}. We are not aware of the stopping contagion problem studied before, nor are we aware of previous studies of the minimizing contagion problem in random graphs. In order to deal with both stopping contagion and minimizing contagion for graphs of bounded treewidth, we build on algorithmic ideas from~\cite{ben2011treewidth}. The NP-hardness of minimizing contagion with respect to vertex deletion is proved in~\cite{cordasco2021parameterized}---our proof for the NP-hardness of the edge deletion version of minimizing contagion was found concurrently and independently; the proof is different from the proof appearing  in~\cite{cordasco2021parameterized}.

There are two regimes of interest for the study of treewidth in sparse random graphs. For the subcritical regime $p \leq d/n$ with $d<1$, $G(n, p)$ has w.h.p. unicyclic connected components of size $O(\log n)$~\cite{erdHos1960evolution} and hence has treewidth at most $2$. For the supercritical regime with $p \geq d/n$ and $d>1$, $G(n,p)$ has w.h.p. a giant component of size $\Omega(n)$~\cite{erdHos1960evolution} and determining the treewidth is more complicated. Kloks~\cite{kloks1994treewidth} proved that the treewidth of $G(n, d/n)$ is $\Omega(n)$ w.h.p. for $d \geq 2.36$. His result was improved by Gao~\cite{gao2012treewidth} who showed that for $d \geq 2.16$, the treewidth of $G(n,d/n)$ is $\Omega(n)$ with high probability. Gao asked if his result can be strengthened to prove that $G(n, d/n)$ has treewidth linear in $n$ w.h.p. for any $d>1$; this was later shown in in~\cite{lee2012rank}. A different and somewhat simplified proof establishing that the treewidth of $G(n,d/n)$ is $\Omega(n)$ w.h.p. was given in~\cite{perarnau2014tree}.  Finally, the fine-grained behavior of treewidth of $G(n, (1+\epsilon)/n)$ was studied in~\cite{do2022note} where it was shown that for sufficiently small $\epsilon$, the treewidth of $G(n, (1+\epsilon)/n)$ is w.h.p.

\begin{align*}
    \Omega\left(\frac{\epsilon^3}{\log{1/\epsilon}}\right)n.
\end{align*}

The first lower bound for the treewidth of random regular graphs appears to be from~\cite{perarnau2014tree}: the authors prove that for every constant $d > d_0$ where $d_0$ is a sufficiently large constant, the treewidth of the random regular graph $G(n,d)$ is $\Omega(n)$ w.h.p. In~\cite{feige2016giant} it was also shown that random graphs with a given degree sequence (with bounded maximum degree) that ensure the existence of a giant component w.h.p. (namely a degree sequence satisfying the Molloy-Reed criterion~\cite{molloy1995critical}) have linear treewidth as well, which implies, using a different arguments from those in~\cite{perarnau2014tree}, that $G(n,d)$ for $d>2$ has linear treewidth w.h.p. A different proof for the linear lower bound of the treewidth of $G(n,d)$ for $d>2$ is given in~\cite{do2022note}.

Several papers have examined notions of local treewidth in devising algorithms for algorithmic problems such as subgraph isomorphism~\cite{eppstein2002subgraph,hajiaghayi2001algorithms,grohe2003local, frick2001deciding}. For example, Grohe~\cite{grohe2003local} defines a graph family $\mathcal{C}$ of having bounded local treewidth if there exists a function $f:\mathbb{N} \rightarrow \mathbb{N}$ such that for every graph $G=(V,E)$ in $\mathcal{C}$ and every integer $r$, for every vertex $v \in V$ the treewidth of the subgraph of $G$ induced on all vertices of distance at most $r$ from $v$ is at most $f(r)$. These works primarily focus on planar graphs and graphs avoiding a fixed minor. The only work we are aware of that has examined the local treewidth of random graphs is that of~\cite{dreier2018local}. Their main goal is to demonstrate that the treewidth of balls of radius $r$ around a given vertex depends only on $r$, as opposed to analyzing the local treewidth as function of $n,d$ and $k$ as we do here. We employ a similar edge excess argument to the one in~\cite{dreier2018local} although there are some differences in the analysis and the results: please see Section~\ref{sec:random} for more details. We are not aware of previous work lower bounding the local treewidth of random graphs. 

Embedding minors in expanders has received attention in combinatorics~\cite{krivelevich2009minors} and theoretical computer science, finding applications in proof complexity~\cite{austrin2022perfect}.
Kleinberg and Rubinfeld~\cite{kleinberg1996short} proved that if $G=(V,E)$ is a $\alpha$-vertex expander with maximum degree $\Delta$, then every graph with $n/\log^{\kappa} n$ vertices and edges is a minor of $G$ for a constant $\kappa>1$ depending on $\Delta$ and $\alpha$. Later it was stated~\cite{chuzhoy2019large} that $\kappa(\Delta, \alpha)=\Omega(\log^2(d)/ \log^2(1/\alpha))$. Krivelevich~\cite{krivelevich2019expanders} together with Nandov proved that if $G$ is an $\alpha$-vertex expander then it contains every graph with $cn/\log n$ edges and vertices for some universal constant $c>0$. 

The sparsity of random graphs as well as randomly perturbed trees was used in showing that these families have w.h.p. bounded expansion\footnote{Bounded expansion should not be confused with the edge expansion of a graph. For a precise definition please see~\cite{nevsetvril2012characterisations,nevsetvril2012sparsity}.}~\cite{nevsetvril2012characterisations,demaine2014structural}. These results are incomparable with our treewidth results: it is known that graphs with bounded maximum degree have bounded expansion and that $G(n,d/n)$ has bounded expansion w.h.p.~\cite{nevsetvril2012characterisations,nevsetvril2012sparsity} In contrast, there exist $3$-regular graphs with linear treewidth and as previously mentioned the treewidth of $G(n,d/n)$ is $\Omega(n)$.

\subsection{Future directions}
Our work raises several questions. We consider undirected unweighted graphs. However directed edges can be more accurate in modeling epidemic spread~\cite{allard2020role} and some edges might be more costly to move than others. 
Extending our algorithms to directed weighted graphs is an interesting direction for future research. 

Our upper and lower bounds for the local treewidth of $G(n,p)$ (with $p=d/n$) currently differ by a multiplicative factor of order $\log d$. We believe that for $k \leq n^{1-\epsilon}$ the local treewdith of $G(n,p)$ is w.h.p. $\Omega(k \log d/\log n)$: whether this is indeed the case is left for future work. Our upper bounds on the local treewidth of noisy trees can be made independent of the maximum degree of the tree; namely, for arbitrary trees, the local treewidth should be upper bounded w.h.p. by $O(k /\log n)$ assuming $k$ is not too large. Proving or disproving this however remains open. Understanding how well one can approximate minimizing contagion and stopping contagion in general graphs, as well as graphs with certain structural properties (e.g. planar graphs) is a potential direction for future research as well. Finally, it could be of interest to study if our bounds for local treewidth coupled with sophisticated algorithms for graphs with bounded local treewidth~\cite{frick2001deciding,grohe2003local,eppstein2002subgraph} could lead to improved running time for additional algorithmic problems in random graphs.
\subsection{Organization}
Our bounds regarding the local treewidth of $G(n,p)$ can be found in Section~\ref{sec:random}. Bounds for the local treewidth of noisy trees can be found in Section~\ref{sec:noisy}. A high-level description of our algorithms for graphs of bounded treewidth can be found in Section~\ref{sec:treewidth}: Complete details and proofs can be found in the Appendix. 
Our algorithms for random graphs and planar graphs can be found in Section~\ref{sec:algorithms}. Our bounds on the local treewidth and algorithms for planar and random graphs are independent from one another: Understanding our algorithms
does not require delving into the the details of the proofs regarding bounds on the local treewidth. Similarly, understating the algorithm for noisy trees (Theorem~\ref{thm:algorthm noisy}) does not require
reading the proofs of any of the Theorems or Lemmas preceding Theorem~\ref{thm:algorthm noisy} in Subsection~\ref{sec:NT}.
\subsection{Preliminaries}~\label{sec:def}
Throughout the paper $\log$ denotes the logarithm function with base $2$; we omit floor and ceiling signs to improve readability. All graphs considered are undirected and have no parallel edges. Given a graph $G=(V,E)$ and two disjoint sets of vertices $A,B$ we denote by $E(A,B)$ the set of edges connecting a vertex in $A$ to a vertex in $B$. For $A,B$ as above we denote by $N_G(A,B)$ the set of vertices in $B$ with a neighbor in $A$. For a subset of vertices $A \subseteq V$ and an edge $e$ we say that $A$ \emph{touches} $e$ if at least one of the endpoints of $e$ belongs to $A$. If both endpoints of $e$ belong to $A$ then we say that $A$ \emph{spans} $e$. 

A graph $H$ is a \emph{minor} of $G$ if $H$ can be obtained from $G$ by repeatedly doing one of three operations: deleting an edge, contracting an edge or deleting a vertex. We keep our graphs simple and remove any parallel edges that may form during contractions. It can be verified~\cite{nevsetvril2012sparsity} that a graph $H$ with $k$ vertices is a minor of $G$
if and only if there are $k$ vertex disjoint connected subgraphs of $G,C_1 \ldots C_k$ such that for every edge $(v_i,v_j)$ of $H$, there is an edge connecting a vertex in $C_i$ to a vertex of $C_j$.
We refer to the map mapping every vertex of $H,v_j$ to $C_j$ as an \emph{embedding} of $H$ in $G$; the maximum vertex cardinality of $C_i,1\leq i \leq k$ is called the \emph{width} of the embedding. We shall be relying on the well-known fact~\cite{nevsetvril2012sparsity,diestel2005graph} that if $H$ is a minor of $G$ than the treewidth of $G$ is lower bounded by the treewidth of $H$. 

We will also need the following definition of an edge expander:

\begin{definition}
Let $\alpha \in (0,1)$. A graph $G$ is an $\alpha$-expander if every subset of vertices $S$ with $|S| \leq n/2$ satisfies $$|N_G(S, V \setminus S)| \geq \alpha|S|.$$ 
\end{definition}

\section{Local treewidth of random graphs}\label{sec:random}
In this section we prove both an upper and lower bound for $t_k(G(n,p))$ that with high probability. We assume $k \leq n^{1-\epsilon}$ for a constant $\epsilon>0$.


\subsection{Upper bound}

Our main idea in upper bounding $t_k(G)$ is to leverage the fact that $G(n,p)$ is locally sparse and that if a few edges are added on top of a tree, the treewidth of the resulting graph cannot grow too much. 

\begin{lemma}\label{lem:excess}
Let $G$ be a connected graph with $n$ vertices and $n-2+\ell$ edges. Then $\tw(G) \leq \ell$.
\end{lemma}
\begin{proof}
Since $G$ is connected, it must have a spanning tree $T$ with $n$ vertices and $n - 1$ edges. The graph $G$ has exactly $\ell - 1$ additional edges; since adding an edge can increase a graph's treewidth by at most $1$, we immediately get the desired bound.
\begin{align*}
    \tw(G) &\le \tw(T) + \ell - 1 = \ell
\end{align*}
\end{proof}
We can now prove:
\begin{theorem}\label{thm:upperbound}
	Suppose that $k\leq n^{1-\epsilon}$. Then for $G = G(n, p)$ we have that w.h.p. for every $m \leq k$:
	\begin{align*}
		t_m(G) \leq 3 + O\left(\frac{m\log d}{\log n} \right).
	\end{align*}
\end{theorem}
\begin{proof}
	Since the Theorem is obvious for $d=n^{\Omega(1)}$ we assume that $d \leq n^{\epsilon/2}$. 
	We first prove the statement for $m=k$.
	Given a graph $G$ with treewidth $t$, it is always possible to find a connected subgraph of $G$ with identical treewidth to $G$. In that spirit, rather than bounding the probability there exists some  $k$-vertex subgraph of $G$ with treewidth exceeding some $r$, we bound the probability some subgraph on $s \le k$ vertices is connected and has treewidth greater than $r$ in $G$.

	Fix some $S \subseteq V$ with exactly $s$ vertices. Note there are $s^{s - 2}$ possible spanning trees which could connect the vertices in $S$, each requiring $s - 1$ edges. While the resulting subgraph would be connected, its treewidth is only $1$. Therefore, $r$ additional edges would also be required to produce a subgraph with treewidth at least $r + 1$. Accounting for the ways to choose these edges, the probability the subgraph induced on $S$ is connected and has treewidth greater than $r$ is at most

	\begin{align*}
		s^{s - 2} \binom{\binom{s}{2}}{r} \left(\frac{d}{n}\right)^{r + s - 1}.
	\end{align*}

	This follows since each edge occurs independently with probability $p=d/n$. Now, we bound the probability that any such subset $S$ with at most $k$ vertices exists. To that end, we take a union bound over all $\binom{n}{s}$ possible subsets of $s$ vertices, letting $s$ range from $1$ to $k$. Putting this together and using the inequality ${a \choose b} \leq (ea/b)^b$ yields

	\begin{align*}
		\sum_{s = 1}^k \binom{n}{s} \times s^{s - 2} \binom{\binom{s}{2}}{r}
	\left(\frac{d}{n}\right)^{r + s - 1}
		&\le \frac{d^r}{n^{r - 1}} \sum_{s = 1}^k e^s \left(\frac{es^2}{2r}\right)^{r} d^s \\
		&\le \frac{d^r}{n^{r - 1}} k e^k \left(\frac{ek^2}{r}\right)^{r} d^k
	\end{align*}

	To complete the proof, notice this probability can be made to be at most $n^{-1}$ (using  $k \leq n^{1-\epsilon}$ and $d \leq n^{\epsilon/2}$) when $r$ is taken to be

	\begin{align*}
		2 + O\left( \frac{k\log d}{\log n} \right).
	\end{align*}
	The Theorem now follows for $m=k$ from Lemma~\ref{lem:excess}.
	Using the above proof along with a simple union bound over all $m \leq k \leq n^{1-\epsilon}$ implies the statement for all $m \leq k$.
\end{proof}

Notice the approach above yields a sharper bound than if we solely attempted to bound the treewidth by counting the number of excess edges above $k - 1$. To explain, notice a $k$-vertex subgraph can have treewidth $r$ only if it has at least $r + k - 1$ edges. A simple union bound over all possible subsets of $k$ vertices, upper bounds the probability we are interested in.

\begin{align*}
	\binom{n}{k} \binom{\binom{k}{2}}{r + k - 1}  \left(\frac{d}{n}\right)^{r + k
- 1} \le \frac{k^{2k} k^{2r} d^k d^r}{n^{r -1}}
\end{align*}

This is implicitly used in~\cite{dreier2018local} to bound the treewidth of balls of radius $r$ in $G(n, p)$; as mentioned above, our method improves on this result. More concretely, since the upper bound now has a additional $k^k$ factor in the numerator, using this in our application would yield the weaker upper bound

\begin{align*}
	t_k(G) = 3 + O\left( \frac{k(\log k + \log d)}{\log n} \right).
\end{align*}

\subsection{Lower bound}

Throughout this section we assume that $d>1+\delta$ where $\delta>0$.

First we need the following result from~\cite{krivelevich2019expanders}:
\begin{prop}\label{prop:kriv}
    Consider the random graph $G:=G(n,\frac{1+\delta}{n}).$ Then there is a constant $c>0$ depending on $\delta$ such that the following holds: For every graph $H$ with at most $k \leq cn/ \log n$ vertices and edges, $G$ contains an embedding $H'$ of $H$. Furthermore the width of the embedding is $O(\log n)$. 
\end{prop}

\begin{prop}\label{prop:linear}
There exist graphs with $m$ vertices and $m$ edges of treewidth $\Omega(m).$
\end{prop}
\begin{proof}
As random $3$-regular graphs have with high probability linear treewidth~\cite{do2022note,feige2016giant} there are $m$-vertex graphs with $m$ vertices and $3m/2$ edges and treewidth $\Omega(m)$. Adding to such a graph $m/2$ isolated vertices results with a graph with the desired property.  
\end{proof}

Using our results we can lower bound the local treewidth of a random graph:
\begin{theorem}\label{thm:lowerbound}
  Let $G:=G(n,d/n)$ be a random graph with $d>1+\delta$. 
  Assume $k \leq O(n/\log n)$. Then w.h.p. $G$ contains a subgraph with $O(k)$ vertices whose treewidth is $\Omega(\frac{k}{\log n}).$
  
\end{theorem}
\begin{proof}
We may assume that $k=\Omega(\log n)$, otherwise the lower bound in the Theorem is immediate.  
Let $H$ be a graph with $s$ vertices and edges of treewidth $\Omega(s)$ (Such a graph exists by proposition~\ref{prop:linear}).
Let $s \leq \left(\frac{n}{\log n}\right)$. 
By Proposition~\ref{prop:kriv}, $G$ contains an embedding of $H,H'$ of width $O(\log n)$. It follows that $H'$ has at most $O(s \log n)$ vertices and treewidth at least $\Omega(s)$ (as $H$ is a minor of $H'$) which is what we wanted to prove. 
\end{proof}

\subsection{Local treewidth of random regular graphs}~\label{sec:random-reg}

Similar bounds on the local treewidth of random regular graphs $G(n,d)$ can be established via similar arguments to those used for $G(n,d/n).$
For the upper bound, one can use the fact~\cite{coja2014contagious} that for every $k<nd/4$ distinct unordered pairs of vertices, the probability they all occur simultaneously in $G(n,d)$
is at most $(2d/n)^k$ and then nearly identical arguments to those in Theorem~\ref{thm:upperbound}. The lower bound follows easily from embedding results for expanders:
\begin{theorem}
    Let $d>2$ be a constant. Then with high probability a random $d$-regular graph $G$ is minor universal: any graph $H$ with at most $O(n/\log n)$ vertices and edges can embedded into $G$.
    Furthermore, the width of the embedding is $O(\log n).$
\end{theorem}
\begin{proof}
    By a result of~\cite{krivelevich2019expanders} if $G$ is an $\alpha$-expander with $\alpha>0$ bounded away from zero then the claim in the Proposition hold. The result now follows as
    it is well known~\cite{bollobas1988isoperimetric,kolesnik2014lower} that with high probability the random $d$-regular graph is an $\alpha$-expander for $\alpha>0$.
\end{proof}

We summarize this with the following Theorem:
\begin{theorem}\label{thm:regular}
	Suppose that $2<d$ is a constant and $k\leq n^{1-\epsilon}$ for some constant $\epsilon \in (0,1)$ . Then for $G = G(n, d)$ we have that w.h.p.:
	\begin{align*}
		\Omega\left(\frac{k}{\log n} \right)\leq t_k(G) \leq 3 + O\left(\frac{k\log d}{\log n} \right).
	\end{align*}
\end{theorem}

\section{Local treewidth of noisy graphs}\label{sec:noisy}

We study the local treewidth of noisy graphs: Recall that in this model there is a base $n$-vertex graph $G$ with maximum degree $\Delta$. On top of this base graph every non edge of $G$ is added independently with probability $1/n$. All proofs missing from this section can be found in the Appendix. Our main result is: 

\begin{theorem}\label{thm:smallworld}
Let $G$ be an $n$-vertex connected graph of maximum degree $\Delta$. Suppose that we add every non-edge of $G$ to $G$ with probability $1 /n$ independently of all other random edges. Call the resulting graph $G'$. With high probability, then, $t_k(G') \leq O(t_k(G) + r)$, where 
\begin{align*}
    r=3 + O\left(\frac{k (\log k+ \log \Delta)}{\log n}\right).
\end{align*}
\end{theorem}

To prove Theorem~\ref{thm:smallworld} we need several Lemmas.
The first is due to~\cite{bagchi2006effect}. While tighter bounds are known~\cite{beveridge1998random},
the simpler bound from~\cite{bagchi2006effect} suffices to establish our asymptotic upper bounds for the local treewidth.

\begin{lemma} \label{lem:bagchi}
	Let $G$ be an $n$-vertex graph of maximum degree $\Delta$. Then the number of connected subgraphs of $G$ with $k$ vertices is at most $n\Delta^{2(k-1)}$.
\end{lemma}

\begin{lemma}
	Suppose that we have a graph $H$ composed of $k$ connected components $C_1, \ldots, C_k$. Suppose that we merge these connected components by adding exactly $k-1$ edges to produce a connected graph $G$. If $t=\max\{ \tw(C_1), \ldots, \tw(C_k) \}$ then the treewidth of $G$ is at most $\max\{t, 1\}$.
\end{lemma}

\begin{proof}
	For connected components $C_1, \dots, C_k$ consider tree decompositions $T_1, \allowbreak \dots, T_k$ with widths at most $t$; these surely exist given the treewidth of each $C_i$ is at most $t$ for all $1 \le i \le k$.

	Assume $C_i$ and $C_j$ are connected in $G$ by an edge from vertex $v_i$ in $C_i$ to $v_j$ in $C_j$. Take the corresponding tree decompositions $T_i$ and $T_j$ and choose arbitrary bags containing $v_i$ and $v_j$ respectively; introduce a new bag $\{v_i, v_j\}$ and connect this to both. Repeating this process for all $k - 1$ edges added to $H$ connects all $T_1, \dots, T_k$. We claim the resulting graph is a valid tree decomposition of $G$ with width at most $t$, proving the proposition.

	Edge counting certifies the resulting connected graph is a tree. Furthermore, every edge in $G$ has some bag containing its endpoints: edges from $H$ have such a bag in some subtree $T_i$, while the remaining edges explicitly have a bag with its endpoints as constructed above. Finally, since new bags introduced share a vertex with each of its two neighbors, subgraphs corresponding to individual vertices remain trees, satisfying the final requirement for tree decompositions.
\end{proof}

Finally we need the following Lemma whose proof follows directly from the handshake Lemma:
\begin{lemma}\label{lem:ave} 
	Let $G$ be a graph with average degree $d$, Then $G$ contains a subgraph $H$ of $G$ with $H$ having minimal degree at least $d/2$.
\end{lemma}

We can prove Theorem~\ref{thm:smallworld}.

\begin{proof}
We may assume $\Delta,k \leq n^{\epsilon}$ for sufficiently small $\epsilon$, as otherwise the upper bound in the Theorem follows from the fact that $t_k(G') \leq k.$
We begin by upper bounding $t_k(G')$ for connected subgraphs of $G'$ of size $k$. Later we show how to lift the connectedness requirement. We consider two possibilities for $G$. In the first, suppose all subgraphs of $G$ on $k$ vertices have at most

\begin{align*}
    \binom{k}{2} - (r + \ell - 1)
\end{align*}

edges. Now fix such a subgraph $H$ of $G$ with $k$ vertices and $\ell \le k$ connected components; we upper bound the probability the corresponding
subgraph $H'$ in $G'$ (the subgraph induced on the vertices of $H$ in $G'$) is connected and has treewidth at least $t_k(G) + r$.

To that end, the probability a fixed set of $\ell - 1$ random edges connect the $\ell$ connected components in $H$ into a single component in $G'$ is $n^{-(\ell - 1)}$. 


By construction, the largest component in $H$ has size no larger than $k - \ell + 1$. Therefore, by our lemma, the treewidth of $H$ together with $\ell - 1$ connecting edges is upper bounded by $t_{k - \ell + 1}(G)$. For $H'$ to additionally have treewidth at least $t_k(G) + r$, a minimum of $r$ additional random edges must be present in $H$ as $t_{k - \ell + 1}(G) \leq t_k(G)$. Therefore, we can upper bound the probability that $H'$ is both connected and has treewidth larger than $r$ by  

\begin{align*}
   \binom{k^2}{r + \ell - 1} n^{-(r + \ell - 1)}
\end{align*}

We count the number of possible subgraphs $H$ with $k$ vertices and $\ell$ connected components. An upper bound can be derived by noticing there are $\binom{k - 1}{\ell - 1}$ ways to choose the positive sizes of the components, denoted $s_1, \dots, s_\ell$. For each set of sizes, we can bound the choices of components using Lemma~\ref{lem:bagchi}:

\begin{align*}
    \prod_{i = 1}^{\ell} n\Delta^{2(s_i -1)} &\le n^\ell \Delta^{2 \sum_{i = 1}^\ell s_i}
= n^{\ell}\Delta^{2k}.
\end{align*}

We can now upper bound the probability there exists some $k$-vertex subgraph of $G$, $H$, whose corresponding subgraph in $G'$ is connected and has treewidth at least $t_k(G) + r$, denoted  $p(n, \Delta, k, r)$. Since $G$ is connected, this is equivalent to the probability $t_k(G') \ge t_k(G) + r$. Now take a union bound over all possible subgraphs $H$ on $k$ vertices and $1 \le \ell \le k$ connected components; in particular, for each possible choice of $\ell$, we multiply an upper bound on the number of possible subgraphs by the maximum probability each subgraph ends up connected and with large treewidth, derived above. Applying the inequality $\binom{q}{s} \le q^s$, we arrive at the following upper bound:

\begin{align*}
    p(n, \Delta, k, r) &\le \sum_{\ell = 1}^k \binom{k - 1}{\ell - 1} n^{\ell}\Delta^{2k} \times \binom{k^2}{r + \ell - 1} n^{-(r + \ell - 1)}
    &\le \frac{k^{2(r+k)}\Delta^{2k}}{n^{r - 1}} \sum_{\ell = 1}^k k^{\ell} \\
    &\le \frac{k^{2r}\Delta^{2k}}{n^{r - 1}} k^{3k+1.}
\end{align*}

Taking logarithms and using our assumptions on $k,\Delta$ we have that for $$r= O\left(\frac{k (\log k+ \log \Delta)}{\log n}\right)$$ it holds that $p(n,\Delta,k,r) \leq O(1/n).$


Now consider the second case where $G$ has some some $k$-vertex subgraph $H$ with $\ell \le k$ connected components and more than

\begin{align*}
    \binom{k}{2} - (r + \ell - 1) \ge {k \choose 2}- 2k + 1
\end{align*}

edges. We may assume without loss of generality that $r \leq k$, as otherwise the inequality
$t_k(G') \leq O(t_k(G)+r)$ trivially holds---in fact, $t_k(G') \leq r$. Given its edge count $H$ has average degree $\Omega(k)$. It follows from
Lemma~\ref{lem:ave} that $H$ contains a subgraph $\tilde{H}$ of minimum degree $\Omega(k)$. Therefore $\tilde{H}$, and hence $H$, have treewidth $\Omega(k)$, since any graph of treewidth $w$ must contain a vertex of degree at most $w$. Hence, we get that $t_k(G') \leq O(t_k(G)+r)$ as $t_k(G') \leq k$; in fact, $t_k(G') \leq O(t_k(G))$ in this case. 

To conclude the proof we need to consider $k$-vertex subgraphs of $G'$ that are not necessarily connected. To prove our claim for subgraphs that are not connected, we need to consider their connected components. As we have shown, the probability there is a connected subgraph with $k'$ vertices for a fixed $k'<k$ with treewidth larger than $t_k'(G)+O(r(k',n,\Delta)$ is $O(1/n)$. Therefore a simple union bounded argument over all $k'\leq k$ (using $k\leq n^{\epsilon}$) as well as the fact that the treewidth of a graph with components $C_1 \ldots C_j$ is $\max(\{\tw(C_1) \ldots \tw(C_j)\})$ concludes the proof. 
\end{proof}
Observe that similarly to the $G(n,d/n)$ case, a simple union bound argument implies that w.h.p $t_s(G') \leq O(t_s(G) + r)$ for all $s \leq k.$ 

Finally, the upper bound in Theorem~\ref{thm:smallworld} is nearly tight for certain noisy trees. 

\begin{theorem}
Consider the $n$ vertex path, $P_n$. Suppose we add every nonedge to $P_n$ with probability $\epsilon/n$ where $\epsilon>0$ is an arbitrary constant.
Call the perturbed graph $P'.$
Then with high probability for any $\Omega(\log n) \leq k \leq O(n/ \log n)$, there exists a subgraph of $P'$ with $O(k)$ vertices with treewidth $\Omega(k/ \log n).$
\end{theorem}
\begin{proof}
Fix $B$ to be a large enough constant. Chop $P_n$ to $n/B$ disjoint paths\footnote{To simplify the presentation we assume $B$ divides $n$. Similar ideas work otherwise.} $A_1 \ldots A_{n/b}$ each of length $B$. Consider now the graph $G$ whose vertex set is $A_1 \ldots A_{n/B}$ and two vertices $A_i$ and $A_j$ are connected if there is an edge (in $P'$) connecting $A_i$ to $A_j$. The probability two vertices in $G$ are connected is at least $$1-\left(1-\epsilon/n\right)^{B^2}\geq \epsilon B^2/2n.$$  

For a fixed graph $H$ with $s$ vertices and edges, it is known~\cite{krivelevich2019expanders} that the supercritical random graph $G(m, \frac{1+\epsilon}{m})$ contains an embedding of $H$ into $G$ as long as $s= O(m/\log m)$. Furthermore the width of the embedding is $O(\log m)$. The probability that two vertices in $G$ are connected is larger than $\frac{1+\epsilon}{n/B}$. Therefore we can embed $H$ into a subgraph $H'$ of $G$ whose size is at most $ s\log n$ such that $H$ is a minor of $H'$. Furthermore as the vertices of $G$ are paths of length $B$ (in $P_n$), the embedding of $H$ into $G$ directly translates to an embedding of $H$ into $P'$ whose width is 
$O(B \log n)=O(\log n)$. Choosing $H$ with $s$ vertices and edges and treewidth $\Omega(s)$ concludes the proof.  
\end{proof}

\section{Algorithms for graphs of bounded treewidth}\label{sec:treewidth}

In this section, we build on the results of ~\cite{cordasco2021parameterized} to provide polynomial time algorithms for bounded treewidth instances of minimizing contagion and stopping contagion. As we sketched in our introduction, we generalize the influence diffusion minimization problem introduced by the authors and use a similar dynamic-programming algorithm. Our main result is the following algorithm for graphs of bounded treewidth $\tau$:

\begin{theorem}\label{thm:treewidth}
 Let $G$ be an $n$ vertex graph with maximum degree $\Delta$, maximum threshold $r$ and treewidth $\tau$. Then both minimizing and stopping contagion can be solved in time $O\left(\tau 1296^\tau \min\{r, \max\{\Delta, 2\} \}^{4\tau} \poly(n) \right)$.
\end{theorem}

For a proof, including a description of our algorithm and runtime analysis, please see the Appendix. Note that to combine subproblems, we must effectively account for the effect of infected vertices elsewhere on each subgraph we consider. We therefore essentially solve minimizing contagion and stopping contagion in a more flexible infection model, where thresholds are allowed to differ between vertices but remain at most $r$; as a result, our theorem cleanly translates to this setting as well.

\section{Algorithms for minimizing and stopping contagion in grids, random graphs and noisy trees}\label{sec:algorithms}
In this section we study how to solve minimizing contagion and stopping contagion when the set of seeds $A$ is not too large and does not spread by too much. We use this along with local treewidth upper bounds to devise algorithms for minimizing and stopping contagion in random graphs. We also consider algorithms for grids and planar graphs. As usual all missing proofs appear in the Appendix. 

Using similar ideas to~\cite{cordasco2021parameterized} (who consider vertex deletions problems) we have the following result for the minimizing contagion problem.  

\begin{theorem}~\label{thm:exp}
  Let $G=(V,E)$ be an $n$-vertex graph. Suppose there are $t$ edges whose removal ensures no more than $r$ vertices are infected in $G$ from the seed set $A \subseteq V$. Then minimizing contagion can be solved optimally in (randomized) $2^{r+t}\poly(n)$ time where $n$ is the number of vertices.
\end{theorem}
 \begin{proof}
 Color every vertex in $V \setminus A$ independently blue or red. Consider a solution of $t$ edges such that after removing these edges
 a set $B$ of cardinality at most $r$ is infected from $A$.
 With probability at least $2^{-r}$ all vertices in $B$ are red.
 For each of the $t$ edges in an optimal solution, exactly one endpoint does not get infected from $A$ as a result of removing the edge.
 With probability at least $2^{-t}$, all these $t$ endpoints belonging to edges in an optimal solution are colored blue.
 Assuming the two events above occur, we run the contagion process only on red vertices and find the set of vertices $B$ infected from $A$.
 Once we recover $B$ we can remove the minimum set of edges from $G$ ensuring only $B$ is infected from $A$.
 Therefore we can solve the problem optimally with probability (at least) $2^{-(t+r)}$. Repeating this process independently $2^{r+t+10}$ times
 results with a randomized algorithm solving this problem with probability at least $2/3$. The running time is $2^{r+t}\poly(n)$ as desired.
 \end{proof}

The algorithm above can become slow if $r$ or $t$ are very large. Additionally, we do not know how to get similar results (e.g., algorithms of running time $2^{|A|}\poly(n)$) for stopping contagion.
Below we show that we can improve upon this algorithm for graphs that have some local sparsity conditions. A key property we use is that for both minimizing contagion and stopping contagion with a seed set $A,$ we restrict our attention to the subgraph of $G$ induced on $\langle A \rangle$.

 \subsection{Grids and planar graphs}
Consider the $n \times n$ grid where all vertices have threshold at least 2 we have the following ``bounded spread" result:
 \begin{lemma}
 In the $n \times n$ grid every set of size $k$ infects no more than $O(k^2)$ vertices.
 \end{lemma}
 \begin{proof}
 Embed the $n \times n$ grid $G=\{1, \ldots, n\}\times \{1, \ldots ,n\}$ in $H=\{0, \ldots, n+1\}\times \{0, \ldots, n+1\}$ in the natural way. 
 Given a subset $A$ of $G$, the \emph{perimeter} of $A$ is the set of all vertices not belonging to $A$ having a neighbor in $A$.
 The crucial observation is that if $A$ is a set of infected seeds, the perimeter of $A$ can never increase during the contagion process~\cite{balogh1998random}.
As the perimeter of $A$ is at most $4k$ the infected set has perimeter at most $4k$ as well. The result follows as every set $A \subseteq \{1, \ldots, n\}\times \{1, \ldots ,n\}$ of size 
$m$ has perimeter $\Omega(\sqrt{m})$. 
 \end{proof}

Using Theorem~\ref{thm:exp} we have that minimizing contagion on the $n$ by $n$ grid with $k=|A|$ can be solved in time $2^{O(k^2)}\poly(n)$. We simply apply the algorithm in Theorem~\ref{thm:exp} to $\langle A\rangle$. Alternatively we can use exhaustive search over all subsets of edges in the graph induced on $\langle A\rangle$ to solve\footnote{For minimizing contagion using the FPT algorithm may be preferable as it may run significantly faster if the optimal solution has cardinality $o(k^2)$.} both minimizing or stopping contagion.
 We can do better using the following fact:
 \begin{lemma}
 Let $G$ be a subgraph of an $n$ by $n$ grid with $r$ vertices. Then $G$ has treewidth $O(\sqrt{r})$.
 \end{lemma}
 \begin{proof}
Every $m$-vertex planar graph has treewidth $O(\sqrt{m})$.
 \end{proof}
We get:
\begin{corollary}
Let $G=(V,E)$ be the $n$ by $n$ grid. Suppose $H=(V,E')$ where $E' \subseteq E$ and every vertex has a threshold of at least $2$. Let $A$ be the seed set with $k=|A|.$
Then stopping contagion and minimizing contagion can be solved in time $2^{O(k)}\poly(n)$.
\end{corollary}
\begin{proof}
For solving either problems we only need to consider the subgraph of $G, \langle A\rangle$. The result now follows from 
Theorem~\ref{thm:treewidth}.
\end{proof}

Similarly, for a planar graph where every vertex has threshold at least $2$ and at most $b$ and every subset $A$ of size $k$ infects at most $f(k)$ vertices, stopping contagion can be solved in time
$b^{O(\sqrt{f(k)})}\poly(n)$. 

\subsection{Sparse random graphs}
Consider the random graph $G(n,d/n)$ assuming all vertices have threshold larger than $1$.
Assuming $d \leq n^{1/2-\delta}$ for $\delta \in (0,1/2)$, it is known~\cite{feige2017contagious} that with high probability every set of size $O(\frac{n}{d^2\log d})$ does not infect more than $O(|A|\log d)$ vertices. Furthermore,
 it is known~\cite{feige2017contagious} that any set of size $O(n/d^2)$ has with high probability constant average degree. It follows that assuming $|A|=O(\frac{n}{d^2\log d })$ the optimal solution to minimizing contagion is of size $O(|A|\log d)$.  Therefore
 in random graphs with $|A| \leq O(\frac{n}{d^2\log d})$, minimizing contagion can be solved using Theorem~\ref{thm:exp} in time $O(2^{|A|\log d}\poly(n))$. As before, exhaustive search over all edges on the graph induced on $\langle A\rangle$ can solve both minimizing and stopping contagion in time $O(2^{|A|\log d}\poly(n))$ as well.
 
Using our local treewidth estimates, Theorem~\ref{thm:treewidth}, the bounded spread property and the fact that w.h.p the maximum degree of $G$ is $O(\log n/ \log \log n)$
we have the following improvement for the running time:
\begin{theorem}
Let $G:=G(n,d/n), \epsilon \in (0,1)$ and $\delta \in (0,1/2)$. Denote by $k$ to be the size of the seed set $A$. Suppose that $k \leq O(\min(n^{1-\epsilon},\frac{n}{d^2\log d}))$ and $d \leq n^{1/2-\delta}$, and that every vertex has threshold larger than $1$. Then w.h.p both minimizing contagion and stopping contagion can be solved in
time $$\exp\left(O\left(\frac{k\log^2 d \log \log n}{\log n}\right)\right)\poly(n).$$
\end{theorem}
\begin{proof}
    As before we can solve either problem on $\langle A\rangle$ using the upper bound on the treewidth from Theorem~\ref{thm:upperbound}, the fact that with high probability $|\langle A\rangle| \leq O(\log d|A|)$ and the algorithm for graphs of bounded treewidth for stopping or minimizing contagion. 
\end{proof}

One can derive similar algorithms for stopping contagion in random $d$-regular graphs (where $d$ is a constant). As the details are very similar to the analysis of the binomial random graph they are omitted.
\subsection{Noisy trees}~\label{sec:NT}

We now devise an algorithm for stopping contagion and minimizing contagion for noisy trees. To achieve this we first prove that for forests every sets of seeds does not spread by much and furthermore this property is maintained after adding a "small" number of edges on top of the edges belonging to the forest. Then we use similar ideas to 
Theorem~\ref{thm:smallworld} and prove that noisy trees are locally sparse in the sense that every subsets of vertices of cardinality $k$ spans w.h.p $k+o(k)$ edges assuming $k$ is not too large. We use this property to prove that any subset $A$ of $k$ seeds infects w.h.p $O(k)$ vertices. Thereafter we can use the algorithms for bounded treewidth to solve either minimizing contagion or stopping contagion on $\langle A \rangle$. We assume throughout this section that $\epsilon \in (0,1)$ is a sufficiently small constant ($\epsilon <1/100$ would suffice for our proofs to go through).

 Let $T$ be an $n$-vertex tree with max degree $\Delta$ and let $T'$ be the noisy tree obtained from $T$.
Here we show that with high probability every set of $k\leq n^{\epsilon}$ seeds does not infect more than $ck$ additional nodes where $c$ is an absolute constant.
One key ingredient is proving that such noisy trees are locally sparse. 

\begin{theorem}\label{thm:spreadtree}
Let $T$ be a tree of maximum degree $\Delta$ and let $T'$ be the noisy graph obtained from $T$.
Suppose $k, \Delta\leq n^{\epsilon}$. Then with high probability every set of $k$ seeds infects no more than $ck$ vertices where $c$ is some absolute constant.
\end{theorem}

We need a few preliminaries before we can prove Theorem~\ref{thm:spreadtree}.

When $T$ is a tree, a subset of vertices with $k$ vertices and $\ell$ components spans exactly $k-\ell$ edges.
A simple modification of the proof of Theorem~\ref{thm:smallworld} yields:

\begin{theorem}\label{thm:edgespan}
Let $T'$ be the noisy tree. Then, with probability $1 - O(1/n)$, every connected subset of vertices of size $k$ in $T'$ spans at most 
\begin{align*}
	(k - 1) + 1 + O\left(\frac{k (\log k+ \log \Delta)}{\log n}\right)
\end{align*}
edges.
\end{theorem}
\begin{proof}
First observe that if $k$ or $\Delta$ are larger than $n^{\Omega(1)}$ the statement in the Theorem follows immediately from the fact that in $G(n,1/n)$ w.h.p every subset $|S|$ of vertices spans at most $2|S|$ edges~\cite{krivelevich2015smoothed}. Hence we shall assume that $\Delta,k \leq n^{\epsilon}.$
	We employ a nearly identical argument to that used in the proof of Theorem~\ref{thm:smallworld}. Here, we are interested in bounding the probability some $k$-vertex subgraph of $T$ ends up connected in $T'$ with at least $k - 1 + r$ edges.
	
	Now fix such a $k$-vertex subgraph of $T$ with $\ell$ connected components, $H$. Since $T$ is a tree, any such subgraph will contain exactly $k - \ell$ edges; $\ell - 1$ random edges will be required to connect $H$ and $r$ additional to reach the $k - 1 + r$ edge threshold. As argued in the proof of Theorem~\ref{thm:smallworld}, this happens with probability upper bounded by
	
	\begin{align*}
   \binom{k^2}{r + \ell - 1} n^{-(r + \ell - 1)}
	\end{align*}
	
    We can now take a union bound over all possible subgraphs (of $T$) on $k$ vertices and $1 \le \ell \le k$ connected components in $T$; this yields the same upper bound on the probability $T'$ contains a connected subgraph with more than $k+r$ edges computed in the proof of Theorem~\ref{thm:smallworld}: 
    
	\begin{align*}
	 \frac{k^{2r}\Delta^{2k}}{n^{r - 1}} k^{3k+1.}
	\end{align*}
	
	Setting $r$ as in the claim ensures the probability some $k$-vertex connected subgraph in $T'$ has at least $k + r$ edges is $O(1/n)$, as desired. 
\end{proof}

We now easily extend this argument to all connected subgraphs of size {\bf at most} $k$ in $T'$, rather than simply those size exactly $k$.
\begin{corollary}
	Let $T, T'$ be as in Theorem~\ref{thm:edgespan}. Suppose $t = n^{\epsilon}$. Then w.h.p every connected subgraph of $T'$ with $k \leq t$ vertices has at most 
	
	\begin{align*}
		(k - 1) + 1 + O\left(\frac{k (\log k+ \log \Delta)}{\log n}\right)
	\end{align*}
	
	edges. 
\end{corollary}

\begin{proof}\label{cor:unionbound}
	From Theorem~\ref{thm:edgespan}, we know the probability that some connected subgraph on $k$ vertices exceeds this number of edges is $O(1/n)$. Taking a union bound over all $k \leq t$ concludes the proof.
\end{proof}

We can now prove:
\begin{theorem}
Let $T,T',k$ be as in Theorem~\ref{thm:edgespan}. Then, w.h.p every subset of vertices of size $k$ in $T'$ spans at most 
\begin{align*}
	(k - 1) + 1 + O\left(\frac{k (\log k+ \log \Delta)}{\log n}\right)
\end{align*}
edges.
\end{theorem}
\begin{proof}
The result Theorem~\ref{thm:edgespan} extends to arbitrary (not necessarily connected) subgraphs of $T'$ by decomposing an arbitrary subgraph with $k$ vertices of $F$ to its $\ell$ connected components with sizes $k_1 + \dots + k_\ell = k$. Applying the bound derived above to each of these connected subgraphs, we know that for some constants $c_1, \dots c_\ell$, we have that with high probability for $C=\max(c_1 \ldots c_\ell)$:

\begin{align*}
	\sum_{i = 1}^\ell \left[ (k_i - 1) + 1 + c_i \frac{k_i (\log k_i+ \log \Delta)}{\log n} \right] &\le k + \frac{C}{\log n} \sum_{i = 1}^\ell k_i \log k_i + C  \frac{k \log \Delta}{\log n} \\
	&\le k + \frac{C}{\log n} k \log k + C  \frac{k \log \Delta}{\log n} \\
	&= (k - 1) + 1 + O \left( \frac{k (\log k + \log \Delta)}{\log n} \right).
\end{align*}

The second inequality follows from rewriting $\log k_i$ as $\log k-\log (k/k_i)$ and collecting the positive terms.
This concludes the proof.
\end{proof}

As before it is easy to extend the Theorem above to all subsets of $T'$ of size at most $k$. Details omitted. 

We proceed to prove the bounded spread property of noisy trees. We need a few more auxiliary Lemmas.
\begin{lemma}\label{lem:spreadtree}
Let $G$ be a forest. Suppose the threshold of every vertex is at least $2$. Then any set of $k$ seeds in $G$ activates less than $k$ additional vertices in $G$.
\end{lemma}
\begin{proof}
If a set of seeds $A$ activates a set $B$ of additional vertices then $|E(A \cup B)| \geq 2|B|$ as every vertex in $B$ must be adjacent to two active vertices.
On the other hand, as $G$ is a forest we have that $|E(A \cup B)|\leq |A|+|B|-1$. Therefore $|B|<|A|$ which is what we wanted to prove.
\end{proof}
 
\begin{lemma}\label{lem:edgeaddition}
Let $G$ be a graph and suppose we add a set of $k$ edges to $G$. Call the resulting graph $H$.
Then $m(H,2) \geq m(G,2)-k$.
\end{lemma}
\begin{proof}
Suppose towards contradiction that $m(H,2)<m(G,2)-k$. Consider a contagious set $A$ in $H$ of minimal size.
$A$ can be turned into a contagious set in $G$ by adding no more than $k$ vertices: we run the contagious process on $G$ and whenever
we reach a vertex that was infected from $A$ (in $H$) because of an additional edge in $H$ we simply add it to $A$. The total number of vertices added in this way is at most $k$. Therefore, if  $m(H,2)<m(G,2)-k$
we would have found a contagious set in $G$ of size smaller than $m(G,2)$ which is absurd.
\end{proof}
Lemmas~\ref{lem:spreadtree} and ~\ref{lem:edgeaddition} easily extend to the case where the threshold of every vertex is at least 2.

We can now prove Theorem~\ref{thm:spreadtree}.

\begin{proof}
We prove the result for the case every vertex has threshold exactly $2$. The result when the thresholds are at least $2$ is similar.
Consider a set $A$ of $k$ seeds. Suppose $A$ infects an additional set of vertices $B$. We now show that w.h.p for a large enough constant $c,$ $|B|$ must be smaller than
$c|A|$. Else, suppose that $A$ infects at least $c|A|$ vertices for some fixed constant $c$. Without loss of generality $A$ infects exactly $c|A|$ additional vertices.
Assuming that $c$ is sufficiency large, we have using Theorem~\ref{thm:edgespan} that w.h.p the number of edges in $T'$ added on top of $F$, the subgraph of $T$ induced on $A \cup B$ is smaller than $(|A|+|B|)/4$. In addition, $F$ 
satisfies by Theorem~\ref{thm:edgespan} the inequality $m(F,2) \geq (|A|+|B|)/2$. Therefore, by Lemma~\ref{lem:edgeaddition}, in $T'$ the subgraph $F'$ induced on $A \cup B$ satisfies w.h.p. $m(F',2) \geq (|A|+|B|)/4$.
On the other hand, we have that $m(F',2) \leq |A|$ as we assume $A$ infects $B$ in $F'$. Taking $c> 3$ leads to a contradiction concluding the proof.
\end{proof}

As the star with $n-1$ leafs shows, the spread of a subset of size $2$ in a noisy tree with degree $\Omega(n)$ can be w.h.p $\Omega(\log n).$ In addition, we believe that this is the worst possible spread: every subset of size $k$ in a noisy tree will not infect with high probability more than $O(k \log n)$ vertices. 
It seems likely that that restriction on $k$ in Theorem~\ref{thm:spreadtree} can be lifted and that the Theorem holds for arbitrary $k$. Whether this is the case is left for future work. 

Finally, we can leverage Theorem~\ref{thm:spreadtree}  to get algorithms for stopping contagion in noisy trees:
\begin{theorem}\label{thm:algorthm noisy}
Let $T$ be a tree and let $T'$ be the noisy tree obtained from $T$. Assume $|A|=k,\Delta \leq n^{\epsilon}$ and that every vertex has threshold larger than $1$. Let $m:=\max(\log \log n,\log \Delta)$. Then both minimizing contagion and stopping contagion can be solved in $T'$ in time
$$\exp\left(O\left(\frac{k(\log k+\log \Delta)m}{\log n}\right)\right)\poly(n).$$
\end{theorem}
\begin{proof}
    This follows the fact that w.h.p $|\langle A \rangle|=O(k)$, the upper bound on the treewidth of the subgraph induced on $\langle A \rangle$ from~\ref{thm:smallworld}, Theorem~\ref{thm:treewidth} and the fact that the maximum degree of $G(n,1/n)$ is $O(\log n/ \log \log n)$ with high probability.
\end{proof}
\section*{Acknowledgements}
We are very grateful to Michael Krivelevich who provided numerous valuable comments and links to relevant work. Josh Erde offered useful feedback. Finally we would like to thank the anonymous referees for helpful comments and suggestions. In particular we thank a reviewer for noting a gap in a claimed proof of a stronger lower bound of $\Omega(k\log d/\log n)$ of the local treewidth of $G(n,d/n)$. 

An inspiration for this paper has been the operation of the Oncology department at Haddasah Ein Karem hospital during the Covid-19 pandemic. Their professionalism and dedication are greatly acknowledged.  
\bibliographystyle{plain}
\bibliography{reference}
\appendix
\section{NP hardness proofs}
Here we present the proofs that stopping contagion and minimizing contagion are NP-hard. We assume throughout this section that every vertex has threshold $r=2$. 
\subsection{Stopping contagion}
We first introduce a slight variant of stopping contagion and prove it is NP-Hard.

\begin{definition}
    In the stopping contagion by vertex deletion problem, we are given as input a graph $G=(V,E)$, two disjoints sets of vertices $A, B \subseteq V$ and a parameter $\ell$. Given that all vertices in $A$ are active, our goal is to determine whether it is possible to ensure no (remaining) vertices from $B$ are infected by deleting at most $\ell$ vertices from $G$. Deleting vertices from $A$ and $B$ is allowed.
\end{definition}

\begin{theorem}\label{thm:vertex_deletion}
    The stopping contagion by vertex deletion problem is NP-Complete. 
\end{theorem}

\begin{proof}
    Clearly, this problem is in NP; we show that is it NP-Hard by reducing it to vertex cover. Recall, for a graph $G = (V, E)$ and a parameter $\ell$, an instance of Vertex Cover asks whether there exists a subset $R \subseteq V$ of size at most $\ell$ such that every edge in $E$ intersects $R$. Given such an instance of Vertex Cover, we construct the following bipartite graph $G' = (A, B, F)$:
   
    \begin{align*}
        A &= V \\
        B &= \{w_{st}: (s,t) \in E)\} \\
        F &= \{ (v, w_{st}) \in V \times E : v = s\ \text{or}\ v = t \}
    \end{align*}
 
    In words, we let $A$ and $B$ represent the vertices and edges of $G$ respectively, connecting each edge in $B$ to both its corresponding endpoints in $A$. Clearly $G'$ can be obtained from $G$ in polynomial time. We now prove there is a vertex cover of size at most $\ell$ in $G = (V, E)$ if and only if it is possible to delete $\ell$ or fewer vertices in $G' = (A, B, F)$ to stop contagion.

    If there is a vertex cover $R$ of size at most $\ell$ in $G$ then deleting $R$ from $G'$ will ensure that no vertex in $B$ is activated, as every node in $B$ has at least one neighbor in $R$ and exactly two neighbors in $A$. On the other hand, suppose we can delete a set of size at most $\ell$ from $G'$ ensuring that no vertex in $B$ is activated after deletions. Notice that without loss of generality, we may assume only vertices in $A$ are deleted. This is because if some node $w_{st} \in B$ is deleted, we could instead delete $s \in A$ while preserving the safety of all vertices in $B$, since $w_{st}$ has exactly two neighbors. Now let $R$ be a set of at most $\ell$ nodes whose deletion from $A$ results with every vertex in $B$ inactive. The point is that $R$ must be a vertex cover in $G$; otherwise, there would an edge $(s,t)$ in $G$ not covered by $R$, which would imply $w_{st} \in B$ will be infected even after $S$ is deleted as both $s$ and $t$ are seeds. This contradiction concludes the proof.

\end{proof}

Sometimes, deleting a vertex from $A$ or $B$ might be impossible (e.g., shutting down a transportation hub that is too central). The problem above remains hard even if we are not allowed to delete vertices from $A \cup B$, although we leave the proof for the supplementary.

\begin{theorem}
    The stopping contagion by vertex deletion problem is NP-Complete, even when the additional constraint is added that only vertices in $V \setminus (A \cup B)$ may be deleted.
\end{theorem}
\begin{proof}
    The proof of Theorem~\ref{thm:vertex_deletion} illustrates that the vertex deletion problem remains NP-Hard even if no deletions of nodes from $B$ are allowed. To show hardness in the case where $A \cup B$ are disallowed, start with $G'$ constructed in the proof. We further modify $G'$ to include $A'=\{s,t\}$ connected to all vertices in $A$. The proof follows by taking the set of seeds to equal $A'$, since these vertices immediately activate all vertices in $A$.
\end{proof}

We can now leverage these results to show hardness results for the edge deletion problem we were interested in.

\begin{theorem}\label{thm:edge_deletion}
    The stopping contagion decision problem is NP-Complete. 
\end{theorem}
\begin{proof}
    Consider an input to the vertex deletion problem. $G=(V,E)$ with disjoint $A, B \subseteq V$ as the set of active protected and nodes respectively. By the proof of Theorem ~\ref{thm:vertex_deletion}, we may assume $V=A \cup B$ and insist only deletions from $A$ are allowed. We transform $G$ to an input graph $G'=(V',E')$ for the edge deletion problem as follows.
    
    We take $V'$ to be the vertex set of $G$, together with a set $Z$ disjoint from $V$ of cardinality $2|A|$. The vertices in $V$ are connected in $G'$ as they are in $G$. The vertices in $Z$ are divided to $|A|$ disjoint pairs: $\{z_{2i-1}, z_{2i}\}$ for $1 \leq i \leq |A|$. Suppose that $A=\{a_i : 1 \le i \le |A|\}.$ Every vertex $a_i \in A$ is connected (only) to $z_{2i}$ and $z_{2i-1}$ in $Z$. No edge in $G'$ connects a vertex in $Z$ to a vertex in $V \setminus A$. Finally, we let the set of active and protected vertices in $G'$ be $A' = Z$ and $B' = B$ respectively. $G'$ can clearly be constructed from $G$ in polynomial time.
    
    We claim that we can stop every vertex in $B$ from being infected (in $G$) from the set of seeds $A$ by deleting at most $\ell$ vertices from $A$ if and only if we can stop the vertices $B$ from getting infected (in $G'$) from the seed set $A$ by deleting at most $\ell$ edges. Indeed, if we can stop contagion in $G$ by deleting $R\subseteq A$ with $|R| \leq \ell$, we can stop contagion in $G'$ by deleting at most $\ell$ edges by simply deleting a single edge from every pair in $Z$ connected to a vertex in $R$. On the other hand, if we can stop contagion in $G'$ by deleting at most $\ell$ edges, we can stop contagion in $G$ by deleting at most $\ell$ vertices in $A$. This is because we can simply delete every vertex in $A$ that is incident (in $G'$) to an edge deleted from $G'$. (Observe that every edge in $G'$ is adjacent to a vertex in $A$.)
    \end{proof}

We conclude this subsection by showing that stopping contagion remains hard even if $|A|=2$ and $|B|=1$.
\begin{theorem}
    The stopping contagion decision problem is NP-Complete when $|A| = 2$ and $|B| = 1$.
\end{theorem}

\begin{proof}
    We will show the $|A| = 2$ and $|B| = 1$ cases are hard separately over the following two lemmas. Taken together, the reductions in Lemma~\ref{lem:small-a} and Lemma~\ref{lem:small-b} naturally compose, which proves our desired result.
\end{proof}

\begin{figure*}[t] \label{fig:reduction}
    \centering
    \includegraphics[width=.6\textwidth]{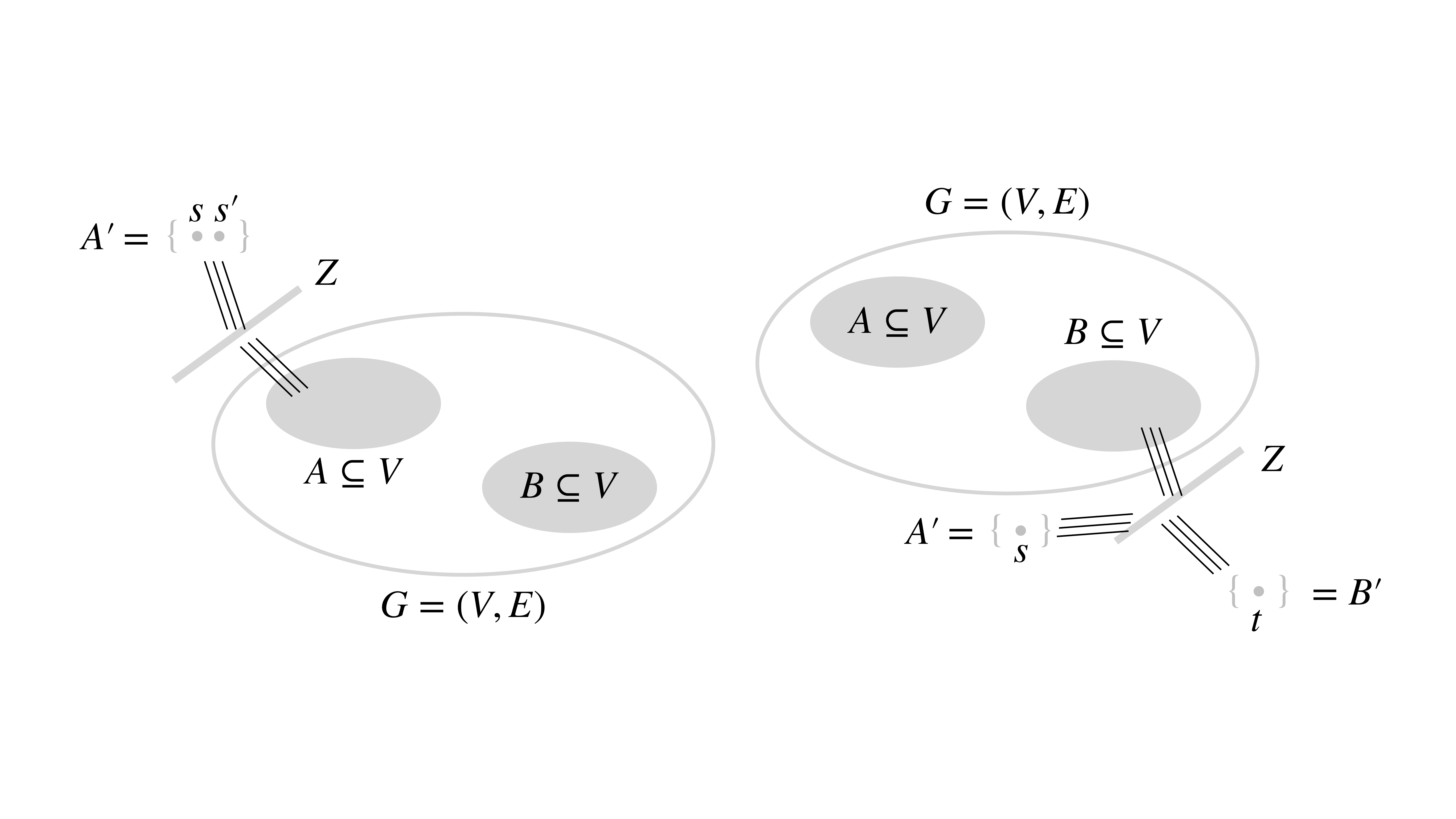}
    \label{fig:first-reduction}
    \caption{$G'$ constructed in reduction to the case where $|A| = 2$ and $|B| = 1$, Lemma~\ref{lem:small-a} and Lemma~\ref{lem:small-b} respectively (left-to-right). In the diagram above, three parallel lines represent a complete bipartite linkage between the two adjacent sets.}
\end{figure*}

\begin{lemma}\label{lem:small-a}
    The stopping contagion decision problem is NP-Complete when $|A| = 2$.
\end{lemma}
\begin{proof}
    Given an instance of the stopping contagion problem on a graph $G = (V, E)$ with $A, B \subseteq V$, we construct a new instance on a graph $G' = (V', E')$ with $A' = \{s, s'\}$ with the same optimum. We do this by attaching a large set of vertices $Z = \{ z_i : 1 \le i \le |E| + 1 \}$ to the graph $G$, and connecting these with $A$ and $A'$. More precisely, we build $G' = (V', E')$ by adding the vertices in $Z$ in the following way.
    \begin{align*}
        V' &= V \cup Z \cup A' \\
        E' &= E \cup (A' \times Z) \cup (Z \times A)
    \end{align*}
    
    Suppose the stopping contagion problem had an optimum of $\ell$ in the original graph; we now argue that $G'$ has the same. First, notice that the optimum is at most $\ell$: deleting the edges which correspond to the original optimal solution will effectively stop the contagion. In the new graph, $A'$ will infect the vertices in $C$ and $A$, but the deleted edges will then prevent the contagion from spreading to $B$.
    
    Now assume, for sake of contradiction, there exists a way to delete $\ell' < \ell$ edges in $G'$ to stop contagion. Given only $\ell'$ edges are deleted, it is impossible to prevent the vertices in $A$ from being infected. This is because each vertex in $A$ is connected to all $|E| + 1$ vertices in $C$, which in turn get infected by $A'$---to prevent this from happening would require deleting at least $|E| \ge \ell > \ell'$ edges. Since $G$ is embedded in $G'$, this would imply that deleting fewer than $\ell$ edges is possible to prevent contagion from spreading from vertices in $A$ to $B$, contradicting the optimality of $\ell$ in the original instance.
    
    This shows that solving the stopping contagion problem in $G'$ can be used to solve the stopping contagion problem in $G$, completing the reduction. Since reduction in the opposite direction is trivial, this shows that the two problems are equivalent under polynomial time reduction. 
\end{proof}

\begin{lemma}\label{lem:small-b}
    The stopping contagion decision problem is NP-Complete when $|B| = 1$.
\end{lemma}
\begin{proof}
    We reduce the general problem of stopping contagion to the case where $|B|=1$. In similar fashion to the previous proof, given an instance of the stopping contagion problem on a graph $G = (V, E)$ with $A, B \subseteq V$, we construct the following graph $G' = (V', E')$ with $B' = \{t\}$.
    \begin{align*}
        V' &= V \cup S \cup Z \cup B' \\
        E' &= E \cup ((B \cup S) \times Z) \cup (Z \times B')
    \end{align*}
    Here, $S = \{s\}$ and $Z$ is an additional large set of $|E| + 1$ vertices introduced to the graph. We now prove that if original instance has optimum $k$, then the stopping contagion problem on $G' = (V', E')$ with $A' = A \cup S$ and $B'$ also has optimum exactly $\ell$.  
    
    It is not hard to see that deleting $k$ edges is sufficient to prevent contagion of $B'$ in $G' = (V', E')$. This follows from the observation that it is possible to remove $\ell$ edges to prevent vertices in $A'$ from contaminating those in $B$: $S$ by itself does not affect the percolation process in the embedded original graph. If $B$ remains inactive, despite the presence of $S$, so will the vertices in $Z$ and thus $B' = \{t\}$.

    To complete the reduction, assume, for sake of contradiction, it was possible to delete $\ell' < \ell$ edges in $G'$ to stop contagion. If a single vertex in $B$ gets infected, so will every vertex in $Z$, thus infecting $B'$ \. Preventing this would require deleting at least $|E| \ge \ell > \ell'$ edges; hence, such a solution could not allow any vertex in $B$ from being infected. Once again, exploiting the fact that $S$ does not affect the percolation process in the embedded graph $G$, this would imply there exists a way to delete at most $\ell'$ edges in $G$ to stop contagion. However, this is a contradiction since $\ell$ is supposed to be the optimum in the original problem, proving our desired result.
\end{proof}

\subsection{Minimizing contagion}

\begin{theorem}
    The minimizing contagion decision problem is NP-Complete.
\end{theorem}
\begin{proof}
    The decision problem is clearly in NP, so we prove it is NP-Hard by reduction from stopping contagion. Recall the in the stopping contagion problem we are given a graph $G = (V, E)$ on $n$ vertices, two disjoint subsets $A, B \subseteq V$ and a parameter $\ell$. The goal is to determine whether $B$ can be protected from infection by deleting at most $\ell$ edges. From Lemma~\ref{lem:small-b}, we can further assume $B = \{b\}$. Given such an instance, we transform it into an input of the minimizing contagion problem as follows.
    
    We create a graph $G'$ from $G$ by adding a set $Z$ of $2n$ vertices. We then choose an arbitrary vertex $a \in A$ and connect all vertices in $Z$ to both $a$ and $b$. Finally, we set the maximum number of infections allowed as $n - 1$ and ask if this possible with $\ell$ edge deletions. Observe that we can assume $\ell \leq n-1$ since we can always prevent $b$ from infection by deleting edges to all other vertices in $G$.
    
    If we can protect $B$ from activation by $A$  in $G$, then deleting the same set of edges from $G'$ ensures that $B$ is not infected from $A$ in $G'$. As a result, no vertex in $Z$ will be infected in $G'$ either, since each has exactly two neighbors: $a$ and $b$. Since $G$ has $n$ vertices, this means that at most $n-1$ vertices can be infected in $G'$. Therefore, if the first problem is solvable with at most $\ell$ edge deletions, so is the minimizing contagion problem on the constructed instance.
    
    To complete the proof, assume we can delete at most $\ell$ edges in $G'$ to prevent more than $n-1$ vertices from being infected. Since $\ell \leq n-1$, at least $n$ vertices in $Z$ must be connected to both $a$ and $b$ even after edges are deleted. It follows that $b$ cannot be infected, lest these $n$ vertices in $Z$ become infected as well. Since $G$ is embedded in $G'$, we have shown that it is possible to prevent $b$ from infection in $G$ by deleting at most $\ell$ edges from $G$. 
    
    Putting these results together, it follows that we can prevent the activation of $b$ in $G$ by deleting at most $\ell$ edges from $G$ if and only if we can ensure at most $n-1$ vertices in $G'$ are activated by deleting at most $\ell$ edges from $G'$. This concludes the proof. 
\end{proof}

\section{Algorithms for graphs of bounded treewidth}
We begin by defining the following problem, which is similar to the Influence Diffusion Minimization problem considered by Cordasco, Gargano \& Rescigno \cite{cordasco2021parameterized}.

\begin{definition}[Generalized Influence Diffusion Minimization]\label{def:GIDM}
    Consider a graph $G = (V, E)$ with corresponding threshold function $t: V \rightarrow \mathbb{N}$ and subsets $A, B \subseteq V$. Given a budget $\ell$, the generalized influence diffusion problem asks to minimize the number of vertices in $B$ infected through setting up to $\ell$ thresholds to $+\infty$ for vertices in $A$.
\end{definition}

This problem uses a more flexible infection model than we initially considered: a vertex $v \in V$ becomes infected after at least $t(v)$ of its neighbors are infected, a threshold which is allowed to differ across vertices. For simplicity, we refer to this infection model as $t$-neighbor bootstrap percolation. Here, the initial set of infected vertices (seeds) are those having threshold $0$. Vertices whose thresholds are set to $+\infty$ are called immunized; these vertices can never be infected.

We now show the generalized influence diffusion problem can be solved efficiently on graphs with bounded treewidth, providing a dynamic programming algorithm. To simply our algorithm, we first introduce a more convenient notion of tree decomposition and note it can also be computed quickly.

\begin{definition}[Nice Tree Decomposition]
	A rooted tree decomposition $(T, X)$ of a graph $G = (V, E)$ is called nice if it satisfies the following additional properties.
\end{definition}

\begin{enumerate}[noitemsep]
	\item If $X_r$ represents the root of $T$, $X_r$ and all leaves are empty.
	\item Every non-leaf node falls into one of the three categories below.
	\begin{enumerate}
		\item Introduce Node: A node $X_i$ with exactly one child $X_j$ such that $X_i = X_j \cup \{ u \}$ for some $u \notin X_j$.
		\item Forget Node: A node $X_i$ with exactly one child $X_j$ such that $X_i = X_j - \{ u \}$ for some $u \in X_j$.
		\item Join Node: A node $X_i$ with exactly two children $X_j$ and $X_k$ with $X_i = X_j = X_k$.
	\end{enumerate}
\end{enumerate}

We need algorithms that compute tree decompositions. It suffices for our purpose to find, given a graph of treewidth $w$ a tree decomposition of width $O(k)$.

\begin{prop}\label{prop:generate-nice}
	Given a graph $G = (V, E)$ on $n$ vertices with treewidth $\tau$, a tree decomposition with width at most $4\tau+3$ can be found in $O(27^{\tau} n^2)$. Let $(T, X)$ be one such tree decomposition with $n'$ nodes; from this, a nice tree decomposition can be computed in time $O(\tau \max\{ n, n'\})$ with width at most $4 \tau + 3$ and number of nodes bounded by $O(\tau n)$.
\end{prop}

The first part of the proposition regarding finding a decomposition of width at most $4\tau+3$ is due to~\cite{robertson1995graph}. The second part regarding nice tree decompositions is from~\cite{kloks1994treewidth}. We are now ready to state our main result.

\begin{prop}\label{prop:gidm}
	Consider an instance of the generalized influence diffusion problem on a graph $G = (V, E)$ with $A, B \subseteq V$, threshold function $t$ and budget $\ell$. If $G = (V, E)$ has treewidth $\tau$, this problem is solvable in time
	\begin{align*}
		O(n^2 27^\tau + \tau n \cdot \ell 81^\tau \min\{r, \Delta\}^{4\tau} \cdot (16^\tau + \ell)) \\
	\end{align*}
	
	Here, $n$ is the number of vertices in $G$ and $\Delta$ its maximum degree. $r$ is the maximum assumed by the threshold function $t$.
\end{prop}
\begin{proof}
	From Proposition~\ref{prop:generate-nice}, a nice tree decomposition of $G = (V, E)$ can be computed efficiently once we obtained our tree decomposition. For this tree decomposition $(T, X)$ with root $X_r$, we introduce the shorthand below to simplify notation.
		
	\begin{itemize}[noitemsep]
		\item $T(i)$ denotes the subtree rooted at $X_i$.
		\item $X(i)$ represents the union of all bags in $T(i)$, including $X_i$ itself.
		\item $G(i)$ denotes the subgraph of $G$ induced by vertices in $X(i)$.
		\item $n_i$ is the cardinality of bag $X_i$.
	\end{itemize}
	
	We now present a dynamic programming algorithm which recursively computes subproblem values at each bag under different hypotheses. To that end, fix a bag $X_i$ and notice after infection spreads each vertex is
	either infected; inactive but not immunized (safe); or immunized and therefore inactive. We denote these states with the letters \texttt{n}, \texttt{s} and \texttt{m} respectively. Let $S_i$ represent the set of all legal mappings $s$ from the vertices in $X_i$ to states. A legal mapping is defined as one where no vertices outside $A$ are mapped to the immunized state; these are basically the mappings which are ``possible.''
	\begin{align*}
		s: X_i \rightarrow \{ \texttt{n}, \texttt{s}, \texttt{m}\}
	\end{align*}
	
	Similarly, let $T'_i$ represent the set of all mappings $t'$ from vertices in $X_i$ to possible threshold values. These mappings will help account for the effect of infected vertices in $V - X(i)$ on the vertices on $X_i$ in $T(i)$: note all edges between $V - X(i)$ and $X(i)$ must involve vertices in $X_i$ in any tree decomposition. As such, for each vertex $u \in X_i$ we only care about thresholds of $0$ up to $t(u)$.	
	\begin{align*}
		t': X_i \rightarrow \mathbb{N},\ t'(u) \in \{0, \dots, t(u)\}
	\end{align*}
	
	Now, denote as $B_i(p, s, t')$ the minimum number of $B$-vertices that can be infected in $G(i)$ by immunizing at most $p$ $A$-vertices in $X(i)$, where the states and thresholds of vertices in $X_i$ are given by $s$ and $t'$ respectively. These are precisely the subproblem values which are computed in a bottom-up fashion at each bag $X_i$, for all $s \in S_i$ and $t' \in T'_i$ and $0 \le p \le \ell$. For bags with no vertices, and hence no legal state or threshold mappings, we aim to compute $B_i(p, \phi, \phi)$ for all $0 \le p \le \ell$, where $\phi$ denotes an empty mapping. (If a bag has no vertices, there is only a single subproblem value to compute for each $p$.)
		
	By construction, $B_r(\ell, \phi, \phi)$ is our desired valued, since $X_r = \phi$. For every leaf $X_i$, it is immediately clear $B_i(p, \phi, \phi) = 0$ for all 
	$0 \le p \le \ell$. With this in place, all that remains to specify the algorithm and prove its correctness is to show how a bag's subproblem values can be correctly computed from those of its children---we give the update rules below, building off those given by 
	
	\begin{itemize}
		\item \emph{Introduce Nodes.} Let $X_i = X_j \cup \{ u \}$ be an introduce node with child $X_j$. For a state mapping $s$ on $X_i$, let $\tilde{s}$ be the corresponding restriction to $X_j$ so $\tilde{s}(v) = s(v)$ for all $v \in X_j$; define $\tilde{t}'$ with respect to $t'$ analogously. For some threshold mapping $\tilde{t}'$ on $X_j$ and set $S \subseteq X_j$, we also define
		
		\begin{align*}
			(\tilde{t}' - S)(v) =
			\begin{cases}
				\tilde{t}'(v) - 1 &: v \in S \\
				\tilde{t}'(v) &: v \notin S.
			\end{cases}
		\end{align*}
		
		In words, $\tilde{t}' - S$ starts with $\tilde{t}'$ and decrements all thresholds of vertices in $S$ by 1. Finally, we denote with $N_i(s)$ the set of infected neighbors of $u$ among vertices in $X_i$ according to some state mapping $s$. Using this new notation, we write the update rule as follows:
		
		\begin{align*}
			B_i(p, s, t') =
			\begin{cases}
				1\{ v \in B\} + \min\limits_{\substack{N \subseteq N_i(s) \\ |N|=t'(u)}} B_j(p, \tilde{s}, \tilde{t}' - (N_i(s) - N)) &: \substack{s(u) = \texttt{n} \\ t'(u) \le |N_i(s)|} \\
				B_j(p, \tilde{s}, \tilde{t}') &: \substack{s(u) = \texttt{s} \\ t'(u) > |N_i(s)|}  \\
				B_j(p - 1, \tilde{s}, \tilde{t}') &: s(u) = \texttt{m} \\
				+\infty &: \text{otherwise.} \\
			\end{cases}
		\end{align*}
		
		\item \emph{Forget Nodes.} Let $X_i = X_j - \{ u \}$ be a forget node with child $X_j$. Here, for state mappings $t'$ and $s$ on $X_i$, denote
		
		\begin{align*}
			\tilde{t}'(v) = \begin{cases}
				\max\{0, t(v) - N_i(s)\} &: v = u \\
				t'(v) &: v \neq u,
			\end{cases}
		\end{align*}
		
		where $N_i(s)$ is defined as above. To compute our desired subproblem problem values, we take the minimum over all, at most three, states $u$ could be in.
				
		\begin{align*}
			B_j(p, s, t') = \min_{\substack{\tilde{s} \in S_j \\ \tilde{s}(v) = s(v) \\ \forall v \in X_i}} B_j(p, \tilde{s}, \tilde{t}') 
		\end{align*}
		
		\item \emph{Join Nodes.} Let $X_i$ be a join node with children $X_j$ and $X_k$ satisfying $X_i = X_j = X_k$. For a state mapping $s$ on $X_i$, let $M(s)$ denote the total number of immunized vertices among those in $X_i$. We can express our update equations as follows:
		
		\begin{align*}
			B_i(p, s, t') = \min_{\substack{p_j + p_k = p + M(s) \\ p_j, p_k \ge M(s)}} \left\{ B_j(p_j, s, t') + B_k(p_k, s, t') \right\}.
		\end{align*}
	\end{itemize}
	
	The correctness of these rules immediately follow from the definition of $B_i(p, s, t')$ given above; this can be formalized through induction. We therefore conclude the proof by providing an analysis of this algorithm's running-time.
	
	First, note the nice tree decomposition provided by Proposition~\ref{prop:generate-nice} has at most $O(\tau n)$ nodes, where $n$ is the number of vertices in $V$. The total subproblem values that must be computed at each bag $X_i$ is either $\ell$ or the product of $\ell$ and the cardinalities of $S_i$ and $T_i$. Notice each vertex has at most 3 legal states, so we can immediately bound the number of mappings in $S_i$ as
	\begin{align*}
		|S_i| \le 3^{n_i} \le 3^{(4\tau + 3) + 1}.
	\end{align*}
	
	In the last step, we use the tree decomposition's width to constrain the size of the bag. To bound the size of $T_i$, consider that under any mapping in $T_i$, $t'(u) \in \{0, \dots, t(u)\}$, where $u$ is some vertex in $X_i$. If $r$ represents the largest threshold under the threshold function given, the image of $t'(u)$ can take on at most $r + 1$ values. Furthermore, since $u$ can have at most $\deg(u)$ infected neighbors, we technically need only consider $\deg(u) + 1$ values of $t'(u)$ from $t(u)$ to $t(u) - \deg(u)$. We can bound this using $\Delta$, the maximum degree of the graph. Combining these observations yields the bound below.
	\begin{align*}
		|T_i| \le \min\{r + 1, \Delta + 1\}^{n_i} \le \min\{r + 1, \Delta + 1\}^{(4\tau + 3) + 1}
	\end{align*}
	
	Computing the update for introduce nodes requires taking a minimum across subsets $N \subseteq N_i(s)$; since $N_i(s)$ has at most $n_i \le (4\tau +3) + 1$ elements, this involves at most $2^{(4\tau +3) + 1}$ terms. Forget and join nodes similarly are similarly expressed as minimums over $O(1)$ and $O(\ell)$ terms respectively. Multiplying the total number of bags, subproblems per bag and work required per subproblem yields the claimed runtime, after accounting for the time to compute the nice tree decomposition as well.
	
	\begin{align*}
		O(n^2 27^\tau + \tau n \cdot \ell 81^\tau \min\{r, \Delta\}^{4\tau} \cdot (16^\tau + \ell))
	\end{align*}
\end{proof}

This result allows us to derive efficient algorithms for minimizing contagion and stopping contagion on bounded treewidth instances.

\begin{prop}\label{prop:min-contagion-tw}
	Consider the minimizing contagion problem on a graph $G = (V, E)$ with $A \subseteq V$ and parameter $k$ under the $r$-neighbor bootstrap percolation infection model. The minimum of the edge deletions required to ensure at most $k$ additional infections occurs can be computed in time
	\begin{align*}
		O\left((n + m)^2 27^\tau + m \cdot \tau (n + m) \cdot m 81^\tau \min\{r, \max\{\Delta, 2\} \}^{4\tau} \cdot (16^\tau + m) \right).
	\end{align*}
	
	Here, $n$ and $m$ are the number of vertices and edges in $G$ respectively while $\Delta$ is the maximum degree of $G$.
\end{prop}
\begin{proof}
	From $G = (V, E)$ we construct a new graph $G' = (V', E')$ by subdividing each edge $e \in E$. Namely, every edge $e=(u,v) \in E$ is replaced with the two edges $(u,w_{u,v})$ and $(w_{u,v},v)$. We can write $V' = V \cup W$, where $W$ represents the set of new vertices in $G'$ which correspond to edges in $G$. Construct the threshold function
	
	\begin{align*}
		t(u) =
		\begin{cases}
			r &: u \in V - A \\
			0 &: u \in A \\
			1 &: u \in W,
		\end{cases}
	\end{align*}
	
	and set $A' = W$, $B' = V - A$. There is a clear correspondence between how infection spreads in $G$ and $G'$ under the $r$-neighbor and $t$-neighbor bootstrap percolation models respectively. Notice that immunizing a threshold $1$ vertex $u \in A'$ adjacent to $v$ and $w$  corresponds to deleting the edge $(v, w)$ in $G$, since infection cannot spread from one neighbor to the other.
	
	Consider solving the generalized influence diffusion problem on the graph $G' = (V', E')$ with $A'$, $B'$ and $t$ as defined above for budget $\ell$. From the relationship highlighted above, this gives the minimum additional infections possible in $G$ using at most $\ell$ edge deletions. To solve the minimizing contagion problem, we can simply try all possible values of $\ell$ and return the smallest value which keeps the number of additional infections below $k$.
	
	Now let the number of vertices and edges in $G$ be $n$ and $m$ respectively. To compute the runtime of this algorithm, we determine the time complexity of solving the generalized influence diffusion problem on $G'$ and multiply this by $m$, the number of values of $\ell$ iterated over. To that end, we attempt to relate properties of $G'$ to those of the input $G$. First, notice subdividing edges of a graph preserves treewidth and
		
	\begin{align*}
		|V'| = n + m,\ |E'| = 2m.
	\end{align*}
	
	The maximum degree of the graph is also very nearly preserved, although some additional degree 2 vertices are now introduced. Finally, the maximum of the threshold function as defined is $r$, assuming $r > 0$. Plugging this into the runtime given in Proposition~\ref{prop:gidm} and including the $m$ pre-factor gives us the desired runtime.
	\begin{align*}
		O\left((n + m)^2 27^\tau + m \cdot \tau (n + m) \cdot m 81^\tau \min\{r, \max\{\Delta, 2\} \}^{4\tau} \cdot (16^\tau + m) \right)
	\end{align*}
\end{proof}

\begin{prop}
	Consider the stopping contagion problem on a graph $G = (V, E)$ with disjoint $A, B \subseteq V$ under the $r$-neighbor bootstrap percolation infection model. Given all vertices in $A$ are infected, the minimum number of edge deletions necessary to ensure that no vertices from $B$ are infected can be computed in time
	\begin{align*}
		O\left((n + m)^2 27^\tau + m \cdot \tau (n + m) \cdot m 81^\tau \min\{r, \max\{\Delta, 2\} \}^{4\tau} \cdot (16^\tau + m) \right).
	\end{align*}
	
	Here, $n$ and $m$ are the number of vertices and edges in $G$ respectively while $\Delta$ is the maximum degree of $G$.
\end{prop}
\begin{proof}
	We take a nearly identical approach to that used in Proposition~\ref{prop:min-contagion-tw}. Once again, we subdivide each edge to obtain $G' = (V', E')$, where $V' = V \cup W$. Here, $W$ represents the set of new vertices in $G'$ which correspond to edges in $G$. Now we define
	
	\begin{align*}
		t(u) =
		\begin{cases}
			r &: u \in V - A \\
			0 &: u \in A \\
			1 &: u \in W,
		\end{cases}
	\end{align*}
	
	and let $A' = W$, $B' = B$. With this threshold function, there is a correspondence between how infection spreads in $G$ and $G'$; this is discussed further in the proof of Proposition~\ref{prop:min-contagion-tw}.
	
	In light of this connection, consider the solution of the generalized influence diffusion problem on the graph $G' = (V', E')$ with $A'$, $B'$ and $t$ as defined above for budget $\ell$. This gives us the number of infected vertices in $B$ if $\ell$ or fewer edges in $G$ are deleted to minimize number of infections in $B$. Therefore, to solve the stopping contagion problem, we can simply try all possible values of $\ell$ and return the smallest value which allows $0$ vertices in $B$ to be infected. 
	
	A similar analysis to the on in the proof of Proposition~\ref{prop:min-contagion-tw} yields an identical runtime, as desired.
	
	\begin{align*}
		O\left((n + m)^2 27^\tau + m \cdot \tau (n + m) \cdot m 81^\tau \min\{r, \max\{\Delta, 2\} \}^{4\tau} \cdot (16^\tau + m) \right)
	\end{align*}
\end{proof}

\end{document}